\documentclass[11pt,letterpaper]{article}

\usepackage{typearea}
\typearea{14}

\usepackage{comment,algorithm,algorithmic,multicol}

\usepackage{wrapfig}
\usepackage{color,colortbl}
\usepackage{float}
\usepackage{amsthm}
\usepackage{amsmath}
\usepackage{amssymb}
\usepackage{graphicx}
\usepackage{xspace}
\usepackage{nicefrac}
\usepackage[colorinlistoftodos,prependcaption,textsize=tiny]{todonotes}

\newcommand{\atodoin}[1]{\todo[linecolor=red,backgroundcolor=green!25,bordercolor=red,inline]{\textbf{AF:~}#1}}

\usepackage{thmtools}
\declaretheorem[name=Lemma]{lemma}

\definecolor{Darkblue}{rgb}{0,0,0.4}
\definecolor{Brown}{cmyk}{0,0.61,1.,0.60}
\definecolor{Purple}{cmyk}{0.45,0.86,0,0}

\definecolor{Darkblue}{rgb}{0,0,0.4}
\usepackage[colorlinks,linkcolor=Darkblue,filecolor=blue,citecolor=blue,urlcolor=Darkblue,pagebackref]{hyperref}
\usepackage[nameinlink]{cleveref}

\newtheorem{theorem}{Theorem}
\newtheorem{corollary}{Corollary}

\newtheorem{claim}{Claim}
\newtheorem{definition}{Definition}

\newtheorem*{hisnote*}{Historical note}
\newtheorem{observation}{Observation}

\newcommand{\namedref}[2]{\hyperref[#2]{#1~\ref*{#2}}}

\newcommand{\lineref}[1]{\namedref{Line}{#1}}

\newcommand{\supp}{{\rm supp}}

\newcommand{\N}{\mathbb{N}}
\newcommand{\R}{\mathbb{R}}

\newcommand{\diam}{{\rm diam}}

\newcommand{\poly}{{\rm poly}}
\newcommand{\est}{{\rm est}}
\newcommand{\mst}{{\rm MST}}

\newcommand{\Texp}{\mathsf{Texp}}

\newcommand{\etal}{{\em et al. \xspace}}

\newcommand{\ddim}{{\rm ddim}}
\newcommand{\eps}{\epsilon}

\newcommand{\rhoR}{\rho_{\phantom{.}_{\hspace{-3pt}\mathcal{R}}}}

\usepackage{pdfsync}

\title{On Strong Diameter Padded Decompositions\thanks{This paper is a full version of the proceedings version \cite{Fil19} published in APPROX 2019. In addition to previously published material, this version contains full details of the core clustering in the cops and robbers algorithm (which somewhat simplifies over \cite{AGGNT19}). In addition, there is a new \Cref{sec:generalGraphs} on general graphs.}}
\author{Arnold Filtser
	\thanks{The reaserch is supported in part by ISF grants No. (1817/17) and (1042/22), by BSF grant No. 2015813, and by the Simons Foundation.}
	\\Bar-Ilan University\\
	Email: \texttt{arnold.filtser@biu.ac.il} }
\date{}

\begin{document}
	\maketitle
	\thispagestyle{empty}
	\nonumber
\begin{abstract}
Given a weighted graph $G=(V,E,w)$, a partition of $V$ is $\Delta$-bounded if the diameter of each cluster is bounded by $\Delta$.
A distribution over $\Delta$-bounded partitions is a $\beta$-padded decomposition if every ball of radius $\gamma\Delta$ is contained in a single cluster with probability at least $e^{-\beta\cdot\gamma}$.
The weak diameter of a cluster $C$ is measured w.r.t. distances in $G$, while the strong diameter is measured w.r.t. distances in the induced graph $G[C]$. The decomposition is weak/strong according to the diameter guarantee.
	
Formerly, it was proven that $K_r$ minor free graphs admit weak decompositions with padding parameter $O(r)$, while for strong decompositions only $O(r^2)$ padding parameter was known.
Furthermore, for the case of a graph $G$, for which the induced shortest path metric $d_G$ has doubling dimension $\ddim$, a weak $O(\ddim)$-padded decomposition was constructed, which is also known to be tight. For the case of strong diameter, nothing was known.

We construct strong $O(r)$-padded decompositions for $K_r$ minor free graphs, matching the state of the art for weak decompositions. Similarly, for graphs with doubling dimension $\ddim$ we construct a strong $O(\ddim)$-padded decomposition, which is also tight. 
We use this decomposition to construct strong $\left(O(\ddim),\tilde{O}(\ddim)\right)$ sparse cover scheme for such graphs.
Our new decompositions and cover have implications to approximating unique games, the construction of light and sparse spanners, and for path reporting distance oracles. 
\end{abstract}

\newpage
\tableofcontents

\newpage
\pagenumbering{arabic}

\section{Introduction}
Divide and conquer is a widely used algorithmic approach. In many  distance related graph problems, it is often useful to randomly partition the vertices into clusters, such that small neighborhoods have high probability of being clustered together.
Given a weighed graph $G=(V,E,w)$, a partitions is $\Delta$-\emph{bounded} if the diameter of every cluster is at most $\Delta$. A distribution $\mathcal{D}$ over partitions is called a $(\beta,\delta,\Delta)$-\emph{padded decomposition}, if every partition is $\Delta$-bounded, and for every vertex $v\in V$ and $\gamma\in[0,\delta]$, the probability that the entire ball $B_G(v,\gamma\Delta)$ of radius $\gamma\Delta$ around $v$ is clustered together, is at least $e^{-\beta\gamma}$ (that is $\Pr[B_G(v,\gamma\Delta)\subseteq P(v)] \ge e^{-\beta\gamma}$, here $P(v)$ is the cluster containing $v$).
If $G$ admits a $(\beta,\delta,\Delta)$-padded decomposition for every $\Delta>0$,
we say that $G$ admits $(\beta,\delta)$-padded decomposition scheme. If in addition $\delta=\Omega(1)$ is a universal constant, we say that $G$ is $\beta$-decomposable. 

A relaxed notion is that of \emph{probabilistic decomposition}, where the guarantee is over pairs rather than balls.
A  distribution $\mathcal{D}$ over partitions is called a $(\beta,\Delta)$-\emph{probabilistic decomposition}, if every partition is $\Delta$-bounded, and for every pair of vertices vertex $u,v\in V$,
$\Pr[P(u)\ne P(v)]\le \beta\cdot\frac{d_G(u,v)}{\Delta}$.
Another relaxation is called \emph{threshold probabilistic decomposition} (abbreviated ThProbabilistic),\footnote{This type of decomposition was also previously called stochastic decomposition \cite{FN22}.} where the success probability is fixed and does not goes to $1$ as the distance between $u$ and $v$ goes to $0$. Formally, a  distribution $\mathcal{D}$ over partitions is called a $(\beta,p,\Delta)$-\emph{ThProbabilistic decomposition}, if every partition is $\Delta$-bounded, and for every pair of vertices vertex $u,v\in V$ at distance at most $\frac\Delta\beta$, $\Pr[P(u)= P(v)]\ge p$.
Similarly to padded decompositions, if $G$ admits a $(\beta,\Delta)$-probabilistic/$(\beta,p,\Delta)$-ThProbabilistic decomposition for every $\Delta>0$, we say that $G$ admits $\beta$-probabilistic/$(\beta,p)$-ThProbabilistic scheme.

\begin{table}[H]
	\begin{tabular}{|l|l|l|}
		\hline
		\textbf{Stochastic Decomposition type} & \textbf{For}                                                                                               & \textbf{Guarantee}                                            \\ \hline
		$(\beta,\delta,\Delta)$-Padded         & $\gamma\in[0,\delta]$, $v\in V$                                                                            & $\Pr[B_G(v,\gamma\Delta)\subseteq P(v)] \ge e^{-\beta\gamma}$ \\ \hline
		$(\beta,\Delta)$-Probabilistic         & $u,v\in V$                                                                                                      & $\Pr[P(u)\ne P(v)]\le \beta\cdot\frac{d_G(u,v)}{\Delta}$      \\ \hline
		$(\beta,p,\Delta)$-ThProbabilistic     &$u,v\in V$, $d_G(u,v)\le \frac{\Delta}{\beta}$ & $\Pr[P(u)= P(v)]\ge p$                                     \\ \hline
\end{tabular}
\end{table}

We will refer to all the three definitions as stochastic decompositions.
Among other applications, stochastic decompositions have been used for multi-commodity flow \cite{KPR93,LR99}, metric embeddings \cite{Bar96,Rao99,Rab03,KLMN04,FRT04,LN05,ABN11,ACEFN20,FL21,Fil22,BFT23}, spanners \cite{HIS13,FN22,HMO21}, edge and vertex cut problems \cite{Mat02,FHL08}, 
distance oracles and routing \cite{AGGM06,MN07,ACEFN20,FL21,Fil22}, near linear SDD solvers \cite{BGKMPT14}, approximation algorithms \cite{CKR04}, spectral methods \cite{KLPT09,BLR10}, and many more.

The \emph{weak} diameter of a cluster $C\subseteq V$ is the maximal distance between a pair of vertices in the cluster w.r.t. the shortest path metric in the entire graph $G$, i.e.  $\max_{u,v\in C}d_G(u,v)$. The \emph{strong} diameter is the maximal distance  w.r.t.~the shortest path metric in the induced graph $G[C]$, i.e. $\max_{u,v\in C}d_{G[C]}(u,v)$.
Stochastic decomposition can be weak/strong according to the provided guarantee on the diameter of each cluster.
It is considerably harder to construct decompositions with strong diameter. Nevertheless, strong diameter is more convenient to use, and some applications indeed require that (e.g. for routing, spanners e.t.c.).

Previous results on stochastic decompositions are presented in \Cref{tab:results}. 
In a seminal work, Klein, Plotkin and Rao \cite{KPR93} showed that every $K_r$ minor free graph admits a weak $\left(O(r^3),\Omega(1)\right)$-padded decomposition scheme.
Fakcharoenphol and Talwar \cite{FT03} improved the padding parameter of $K_r$ minor free graphs to $O(r^2)$ (weak diameter).
Finally, Abraham, Gavoille, Gupta, Neiman, and Talwar \cite{AGGNT19} improved the padding parameter to $O(r)$, still with weak diameter.
The first result on strong diameter for $K_r$ minor free graphs is by Abraham, Gavoille, Malkhi, and Wieder \cite{AGMW10}, who showed that $K_r$ minor free graph admit a strong $2^{-O(r)}$-probabilistic decomposition scheme.
% In fact, they study a somewhat weaker notion of decomposition called separating decompositions (see \Cref{def:separating}).
Afterwards, Abraham \etal \cite{AGGNT19} (the same paper providing the state of the art for weak diameter), proved that $K_r$ minor free graphs admit strong $(O(r^2),\Omega(\frac{1}{r^2}))$-padded decomposition scheme.
It was conjectured by \cite{AGGNT19} that $K_r$ minor free graphs admit padded decomposition scheme with padding parameter $O(\log r)$. However, even improving the padding parameter for the much more basic case of small treewidth graphs remains elusive. 
%It is also remained open to improve the padding parameter of strong diameter decompositions to match the state of the art of weak diameter.

Another family of interest are graphs with bounded doubling dimension\footnote{A metric space $(X, d)$ has doubling dimension $\ddim$ if every ball of radius $2r$ can be 	covered by $2^{\ddim}$ balls of radius $r$. The doubling dimension of a graph is the doubling dimension of its induced shortest path metric.}.
Abraham, Bartal and Neiman \cite{ABN11} showed that a graph with doubling dimension $\ddim$ is weakly $O(\ddim)$-decomposable, generalizing a result from \cite{GKL03}.
No prior strong diameter decomposition for this family is known.

General $n$-vertex graphs admit strong $O(\log n)$-probabilistic decomposition scheme as was shown by Bartal \cite{Bar96}. In fact, the same proof can also be used to show that general graphs admit strong $\left(O(\log n),\Omega(1)\right)$-padded decomposition scheme. Both results are asymptotically tight. Stated alternatively, there is some constant $c$, such that for every $\Delta>0$, and $k\ge 1$, every $n$-vertex graph admits a distribution over $\Delta$ bounded partitions, such that every pair of vertices at distance at most $\frac{\Delta}{ck}$ will be clustered together with probability at least $n^{-\frac1k}$. 
As the padding parameter governs the exponent in the success probability, it is important to optimize the constant $c$. 
It follows implicitly from the work of Awerbuch and Peleg \cite{AP90} that for $k\in\N$, general $n$-vertex graphs admit a strong $(4k-2,O(\frac1k\cdot n^{-\frac1k}))$-ThProbabilistic decomposition scheme. See the introduction to \Cref{sec:generalGraphs} for additional details.

A related notion to padded decompositions is \emph{sparse cover}.
A collection $\mathcal{C}$ of clusters is a weak/strong $(\beta,s,\Delta)$ sparse cover if it is weakly/strongly $\Delta$-bounded, each ball of radius $\frac\Delta\beta$ is contained in some cluster, and each vertex belongs to at most $s$ different clusters. 
A graph admits weak/strong $(\beta,s)$ sparse cover scheme if it admits weak/strong $(\beta,s,\Delta)$ sparse cover for every $\Delta>0$.
Awerbuch and Peleg \cite{AP90} showed that for $k\in\N$, general $n$-vertex graphs admit a strong $(4k-2,2k\cdot n^{\frac1k})$ sparse cover scheme. 
For $K_r$ minor free graphs, Abraham et al. \cite{AGMW10} constructed strong $(O(r^2),2^r(r+1)!)$ sparse cover scheme. Busch, LaFortune and Tirthapura \cite{BLT14} constructed strong $(4,f(r)\cdot \log n)$ sparse cover scheme for $K_r$ minor free graphs.\footnote{$f(r)$ is a function coming from the Robertson and Seymour structure theorem \cite{RS03}.}

For the case of graphs with doubling dimension $\ddim$, Abraham et al. \cite{AGGM06} constructed a strong $(4,4^{\ddim})$ sparse cover scheme.
No other tradeoffs are known. In particular, if $\ddim$ is larger than $\log\log n$, the only way to get a sparse cover where each vertex belongs to $O(\log n)$ clusters is through \cite{AP90}, with only $O(\log n)$ padding.
 
\subsection{Results and Organization}
In our first result (\Cref{thm:StrongMinorFree} in \Cref{sec:MinorFree}), we prove that $K_r$ minor free graphs are strongly $(O(r),\Omega(\frac{1}{r}))$-decomposable. Providing quadratic improvement compared to \cite{AGGNT19}, and closing the gap between weak and strong padded decompositions on minor free graphs.

Our second result (\Cref{cor:PaddedDoubling} in \Cref{sec:doubling}) is the first strong diameter padded decompositions for doubling graphs, which is also asymptotically tight. Specifically, we prove that graphs with doubling dimension $\ddim$ are strongly $O(\ddim)$-decomposable. 

Both of these padded decomposition constructions are based on a technical theorem (\Cref{thm:padded} in \Cref{sec:centersToPadded}). Given a set of centers $N$, such that each vertex has a center at distance $\le\Delta$, and at most $\tau$ centers at distance $\le 3\Delta$ ($\forall v,~|B_G(v,3\Delta)\cap N|\le \tau$), we construct a strong $(O(\log\tau),\Omega(1),4\Delta)$-padded decomposition.
%We also provide an alternative construction for the decomposition of \Cref{thm:padded} in \Cref{sec:cones}.
All of our decompositions can be efficiently constructed in polynomial time.

Our third contribution is in \Cref{sec:generalGraphs}, and it is for general graphs. Here we construct strong $(2k, \Omega(n^{-\frac{1}{k-1}}))$-ThProbabilistic decomposition scheme (\Cref{thm:PartitionStrongMetrics}),
and weak $(2k, n^{-\frac{1}{k}})$-ThProbabilistic decomposition scheme (\Cref{thm:PartitionWeakMetrics}).
This bounds are close to being tight, assuming Erd\H{o}s girth conjecture \cite{Erdos64}, as for $t<2k+1$, if every $n$ point metric space admits a weak $(t,p)$-ThProbabilistic decomposition scheme, then $p=\tilde{O}(n^{-\frac{1}{k}})$ (\Cref{thm:PartitionMetricsLB}).
Note that this implies that the success probability in our decompositions cannot be improved. However, it might be possible to improve the stretch parameter to $2k-1$ without changing the success probability.
See \Cref{tab:results} for a summery of results on stochastic decompositions.

\begin{table}[t]\center
	\begin{tabular}{|l|l|l|l|l|}
		\hline
		\textbf{Partition type} & \textbf{Diameter} & \textbf{Padding/stretch} & \textbf{~$\delta$ or $p$}       & \textbf{Ref/Notes}\\ \hline
		\multicolumn{5}{|c|}{\textbf{$K_r$ minor free}} \\ \hline
		Padding& Weak            & $O(r^3)$         & $\Omega(1)$        		& \cite{KPR93} \\ \hline
		Padding& Weak            & $O(r^2)$         & $\Omega(1)$       		& \cite{FT03}   \\ \hline
		Padding& Weak            & $O(r)$           & $\Omega(1)$        		& \cite{AGGNT19}   \\ \hline	
		Probabilistic& Strong       & $\exp(r)$        & n/a	         		& \cite{AGMW10}   \\ \hline	
		Padding& Strong          & $O(r^2)$         & $\Omega(\frac{1}{r^2})$   &\cite{AGGNT19}   \\ \hline	
		\rowcolor[HTML]{EFEFEF}  Padding& Strong          & $O(r)$           & $\Omega(\frac{1}{r})$     &  \Cref{thm:StrongMinorFree}\\ \hline
		
		\multicolumn{5}{|c|}{\textbf{Graphs with doubling dimension $\ddim$}} \\ \hline	
		Padding& Weak            & $O(\ddim)$           & $\Omega(1)$           &  \cite{GKL03,ABN11} \\\hline
		\rowcolor[HTML]{EFEFEF} Padding& Strong          & $O(\ddim)$           & $\Omega(1)$           & \Cref{cor:PaddedDoubling} \\ \hline
		
		\multicolumn{5}{|c|}{\textbf{General $n$-vertex graphs}} \\ \hline
		ThProbabilistic, $k\in\N$		& Strong          	& $4k-2$          & $\Omega(\frac{1}{k}\cdot n^{-\frac1k})$           & \cite{AP90}, implicit \\ \hline 
		Padded & Strong          & $O(\log n)$          & $\Omega(1)$	          & \cite{Bar96} \\ \hline 
%		Padding, $k\in\N$		& Weak          	& $O(k)$          & $\Omega(n^{-\frac1k})$           & \cite{MN07} \\ \hline 
		\rowcolor[HTML]{EFEFEF} ThProbabilistic, $k\in\R_{\ge1}$	& Strong          	& $2k$            & $\Omega(n^{-\frac{1}{k-1}})$           & \Cref{thm:PartitionStrongMetrics}\\ \hline 
		\rowcolor[HTML]{EFEFEF} ThProbabilistic, $k\in\N$		& Weak          	& $2k$        & $n^{-\frac1k}$           & \Cref{thm:PartitionWeakMetrics} \\ \hline 
	\end{tabular}
	\caption{Summery of all known and new results on stochastic decomposition schemes for various graph families. The $\delta$ or $p$ column represents the $\delta$ parameter in stochastic decompositions, and the success probability ($p$) in ThProbabilistic decompositions (it is not applicable in probabilistic decompositions).
%	Following the lines of the proof in \cite{Bar96},  one can actually obtain a strong $\left(O(\log n),\Omega(1)\right)$-padded decomposition scheme.
	}\label{tab:results}
\end{table}

%\atodoin{We should somehow decide on the naming here. As we actually have three types.
%Probabilistic appeared in \cite{AGMW10,Bar96}. Some contain padding.
%Stochastic was the name used in \cite{FN22}}

Our fourth result (\Cref{thm:DdimCover} in \Cref{sec:doubling}) is a strong sparse cover for doubling graphs. For every parameter $t\ge 1$, we construct a strong $\left(O(t),O(2^{\nicefrac{\ddim}{t}}\cdot\ddim\cdot\log t)\right)$ sparse cover scheme. Note that for $t=1$ we (asymptotically) obtain the result of \cite{AGMW10}. However, we also get the entire spectrum of padding parameters. In particular, for $t=\ddim$ we get a strong $\left(O(\ddim),\tilde{O}(\ddim)\right)$ sparse cover scheme. 
 
Next, we overview some of the previously known applications of strong diameter stochastic decomposition, and analyze the various improvements achieved using our results. Specifically:

\begin{enumerate}
	\item Given an instance of the unique games problem where the input graph is $K_r$ minor free, 
	Alev and Lau \cite{AL17} showed that if there exist an assignment that satisfies all but an $\eps$-fraction of the edges, then there is an efficient algorithm that finds an assignment that satisfies all but an $O(r\cdot\sqrt{\eps})$-fraction. Using our padded decompositions for minor-free graphs, we can find an assignment that satisfies all but an $O(\sqrt{r\cdot\eps})$-fraction of the edges. See \Cref{subsec:UG}.
	
	\item Using the framework of Filtser and Neiman \cite{FN22}, given an $n$ vertex graph, with doubling dimension $\ddim$, for every parameter $t>1$ we construct a graph-spanner with stretch $O(t)$, lightness $O(2^{\frac{\ddim}{t}}\cdot t\cdot\log^2n)$ and $O(n\cdot 2^{\frac{\ddim}{t}}\cdot\log n\cdot \log \Lambda)$ edges\footnote{Lightness is the ratio between the weight of the spanner to the weight of the MST.\\ ${\Lambda=\nicefrac{\max_{u,v\in V}d_G(u,v)}{\min_{u,v\in V}d_G(u,v)}}$ is the aspect ratio.\label{foot:aspect}}. 
	The only previous spanner of this type appeared in \cite{FN22}, was based on weak diameter decompositions, had the same stretch and lightness, while having no bound whatsoever on the number of edges. See \Cref{subsec:Spanner}.
	
	\item Elkin, Neiman and Wulff-Nilsen \cite{ENW16} constructed a path reporting distance oracle for $K_r$ minor free graphs with stretch $O(r^2)$, space $O(n\cdot\log\Lambda\cdot\log n)$ and query time $O(\log\log\Lambda)$. That is, on a query $\{u,v\}$ the distance oracle returns a $u-v$ path $P$ of weight at most $O(r^2)\cdot d_G(u,v)$ in $O(|P|+\log\log\Lambda)$ time.
	Using our strong diameter padded decompositions we improve the stretch to $O(r)$, while keeping the other parameters intact. See \Cref{subsec:DO}.
	
	\item We further use the framework of \cite{ENW16} to create a path reporting distance oracle for  graphs having doubling dimension $\ddim$ with stretch $O(\ddim)$, space $O(n\cdot\ddim\log\Lambda)$ and query time $O(\log\log\Lambda)$. This is the first path reporting distance oracle for doubling graphs. The construction uses our strong sparse covers. See \Cref{subsec:DO}.
\end{enumerate}

\subsection{Related Work}
%Other than padded decompositions, separating decompositions have been studied. Here, instead of analyzing the probability to cut a ball, we analyze the probability to cut an edge \cite{Awe85,LS93,CKR04,FRT04}.
%Separating decompositions been used to minimize the number of inter-cluster edges in a partition. In particular, strong diameter version of such partitions were used for SDD solvers \cite{BGKMPT14}.

Miller et al. \cite{MPX13} constructed strong diameter partitions for general graphs, which they later used to construct spanners and hop-sets in parallel and distributed regimes (see also \cite{EN19}). 
Hierarchical partitions with strong diameter had been studied and used for constructing distributions over spanning trees with small expected distortion \cite{EEST08,AN19}, Ramsey spanning trees \cite{ACEFN20}, spanning clan embeddings \cite{FL21}, and for spanning universal Steiner trees \cite{BDRRS12}. 
Another type of diameter guarantee appearing in the literature is when we require only weak diameter, and in addition for each cluster to be connected \cite{EGKRTT14,FKT19,Fil20}.

Related type of partitions are sparse partitions \cite{JLNRS05} which is a single partition into bounded diameter clusters such that every small ball intersect small number of clusters (in the worst case), and scattering partitions \cite{Fil20} where the guarantee is that every shortest path intersect small number of clusters (worst case).

Stochastic decompositions were studied for additional graph families. 
Abraham \etal \cite{AGGNT19} showed that pathwidth $r$ graphs admit strong $\left(O(\log r),\Omega(1)\right)$-padded decomposition scheme, implying that treewidth $r$ graphs admit strong $\left(O(\log r+\log\log n),\Omega(1)\right)$-padded decomposition scheme (see also \cite{KK17}). In addition, \cite{AGGNT19} showed that treewidth $r$ graphs admit strong  $(O(r),\Omega(\frac1r))$-padded decomposition scheme. 
Finally \cite{AGGNT19} proved that genus $g$ graphs admit strong $O(\log g)$-padded decomposition scheme, improving a previous weak diameter version of Lee and Sidiropoulos \cite{LS10}.

\subsection{Follow up work}
Following the techniques developed in this paper, in a companion paper \cite{Fil20} we constructed strong diameter sparse partition. Recently, this techniques were pushed further to construct hierarchical strong sparse partition \cite{BCFHHR23}. This finally led to a poly-logarithmic stretch solution for the universal Steiner tree problem (UST).
Sparse partitions were also recently used to solve the facility locations problem in high dimension in the streaming model \cite{CFJKVY23}.

Our sparse covers for doubling metrics (\Cref{thm:DdimCover}), and stochastic decompositions for genral metrics (\Cref{thm:PartitionWeakMetrics}) have been recently used to construct ultrametric covers, which in turn were used to construct reliable spanners 
\cite{FL22,Fil23,FGN23}.

%There been recently developed a new decomposition for doubling metrics with padding parameter $O(\ddim)$ that can be executed in an online fashion \cite{BFT23}. However, this decomposition inherently has weak diameter.

\subsection{Technical Ideas}
The basic approach for creating padded decompositions is by ball carving \cite{Bar96,ABN11}. That is, iteratively create clusters by taking a ball centered around some vertex, with radius drawn according to exponential distribution. The process halts when all the vertices are clustered. 
Intuitively, if every vertex might join the cluster associated with at most $\tau$ centers, the padding parameter is $O(\log \tau)$. We think of these centers as \emph{threateners}. This approach worked very well for general graphs as the number of vertices is $n$. Similarly it also been used for doubling graphs \cite{BFT23}, where the number of threateners is bounded by $2^{O(\ddim)}$. However, in doubling graphs ball carving produces only weak diameter clustering.

Our main technical contribution is a proof that the intuition above holds for strong diameter as well. Specifically, we show that if there is a set $N$ of centers such that each vertex has a center at distance at most $\Delta$, and at most $\tau$ centers at distance $3\Delta$ (these are the threateners), then the graph admits a strong $\left(O(\log\tau),\Omega(1),4\Delta\right)$-padded decomposition scheme. We use the clustering approach of Miller \etal \cite{MPX13} with exponentially distributed starting times. 
In short, in \cite{MPX13} clustering, each center $x$ samples a starting time $\delta_x$. A vertex $v$ joins the cluster of the center $x_i$ maximizing $\delta_x-d_G(x,v)$. This approach guaranteed to creates strong diameter clusters. The key observation is that if for every center $y\ne x_i$, $\left(\delta_{x_i}-d_G(x_i,v)\right)-\left(\delta_y-d_G(y,v)\right)\ge 2\gamma \Delta$, then the ball $B_G(v,\gamma\Delta)$ is fully contained in the cluster of $x_i$. Using truncated exponential distribution, we lower bound the probability of this event by $e^{-\gamma\cdot O(\log \tau)}$. It is the first time \cite{MPX13}-like algorithm is used to create padded decompositions.
%
%In addition to the \cite{MPX13}-based algorithm, we also show a simpler algorithm, based on cone carving (\cite{EEST08}). The cone approach, although less involved, is inherently sequential and implies dependencies of each vertex on the entire center set. 
The \cite{MPX13} algorithm can be efficiently implemented in distributed and parallel settings. Moreover, as each vertex depends only on centers in its local area, we are able to use the  Lov\'asz Local Lemma to create a sparse cover from padded decompositions.

Decompositions of $K_r$ minor free graphs did not use ball carving directly. Rather, they tend to use the topological structure of the graph. We say that a cluster of $G$ has an $r$-core with radius $\Delta$ if it contains at most $r$ shortest paths (w.r.t. $d_G$) such that each vertex is at distance at most $\Delta$ from one of these paths.
%Abraham \etal 
\cite{AGGNT19}'s strong decomposition for $K_r$ minor free graphs is based on a partition into  $1$-core clusters, such that a ball with radius $\gamma\Delta$ is cut with probability at most $1-e^{-O(\gamma r^2)}$. This partition is the reason for their $O(r^2)$ padding parameter. Although not stated explicitly, \cite{AGGNT19} also constructed a partition into $r$-core clusters, such that a ball with radius $\gamma\Delta$ is cut with probability at most $1-e^{-O(\gamma r)}$. Apparently, \cite{AGGNT19} lacked an algorithm for partitioning $r$-clusters.
Taking a union of $\Delta$-nets from each shortest path to the center set $N$, it will follow that each vertex has at most $O(r)$ centers in its $O(\Delta)$ neighborhood. In particular, our theorem above implies a clustering of each $r$-core cluster into bounded diameter clusters. Our strong decomposition with parameter $O(r)$ follows.

We include a new and full proof for the $r$-core clusters. See \Cref{subsec:coreClustering} and the discussion therein.
\section{Preliminaries}

\paragraph{Graphs.}
We consider connected undirected graphs $G=(V,E)$ with edge weights
$w: E \to \R_{\ge 0}$. We say that vertices $v,u$ are neighbors if $\{v,u\}\in E$. Let $d_{G}$ denote the shortest path metric in
$G$.
$B_G(v,r)=\{u\in V\mid d_G(v,u)\le r\}$ is the closed ball of radius $r$ around $v$. For a vertex $v\in
V$ and a subset $A\subseteq V$, let $d_{G}(x,A):=\min_{a\in A}d_G(x,a)$,
where $d_{G}(x,\emptyset)= \infty$. For a subset of vertices
$A\subseteq V$, $G[A]$ denotes the induced graph on $A$,
and $G\setminus A := G[V\setminus A]$.

\begin{wrapfigure}{r}{0.2\textwidth}
	\begin{center}
		\vspace{-20pt}
		\includegraphics[width=0.18\textwidth]{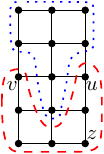}
		\vspace{-7pt}
	\end{center}
	\vspace{-10pt}
\end{wrapfigure}
The \emph{diameter} of a graph $G$ is $\diam(G)=\max_{v,u\in V}d_G(v,u)$, i.e. the maximal distance between a pair of vertices.
Given a subset $A\subseteq V$, the \emph{weak}-diameter of $A$ is $\diam_G(A)=\max_{v,u\in A}d_G(v,u)$, i.e. the maximal distance between a pair of vertices in $A$, w.r.t. to $d_G$. The \emph{strong}-diameter of $A$ is $\diam(G[A])$, the diameter of the graph induced by $A$. For illustration, in the figure to the right, consider the lower cluster encircled by a dashed red line. The weak-diameter of the cluster is $4$ (as $d_G(v,z)=4$) while the strong diameter is $6$ (as $d_{G[A]}(v,u)=6$).
A graph $H$ is a \emph{minor} of a graph $G$ if we can obtain $H$ from
$G$ by edge deletions/contractions, and isolated vertex deletions.  A graph
family $\mathcal{G}$ is \emph{$H$-minor-free} if no graph
$G\in\mathcal{G}$ has $H$ as a minor.
Some examples of minor free graphs are planar graphs ($K_5$ and $K_{3,3}$ minor-free), outer-planar graphs ($K_4$ and $K_{3,2}$ minor-free), series-parallel graphs ($K_4$ minor-free) and trees ($K_3$ minor-free).

\paragraph{Doubling dimension.}  The doubling dimension of a metric space is a measure of its local ``growth rate''. 
A metric space $(X,d)$ has doubling constant $\lambda$ if for every $x\in X$ and radius
$r>0$, the ball $B(x,2r)$ can be covered by $\lambda$ balls of radius $r$. The doubling dimension is defined as $\ddim=\log_2\lambda$. A $d$-dimensional $\ell_p$ space has $\ddim=\Theta(d)$, and every $n$ point metric has $\ddim\le \log n$.
We say that a weighted graph $G=(V,E,w)$ has doubling dimension $\ddim$, if the corresponding shortest path metric $(V,d_G)$ has doubling dimension $\ddim$.
The following lemma gives the standard packing property of doubling metrics (see, e.g., \cite{GKL03}).
\begin{lemma}[Packing Property] \label{lem:doubling_packing}
	Let $(X,d)$ be a metric space  with doubling dimension $\ddim$.
	If $S \subseteq X$ is a subset of points with minimum interpoint distance $r$ that is contained in a ball of radius $R$, then 
	$|S| \le \left(\frac{2R}{r}\right)^{O(\ddim)}$.
%	\atodo{If we really try to optimize, then $|S|\le2^{\ddim\cdot\left\lceil \log\frac{2R}{r}\right\rceil }$}
\end{lemma}

\paragraph{Nets.} A set $N\subseteq V$ is called a $\Delta$-net, if for every vertex $v\in V$ there is a net point $x\in N$ at distance at most $d_G(v,x)\le \Delta$, while every pair of net points $x,y\in N$, is farther than $d_G(x,y)>\Delta$.
A $\Delta$-net can be constructed efficiently in a greedy manner. In particular, by \Cref{lem:doubling_packing}, given a $\Delta$-net $N$ in a graph of doubling dimension $\ddim$, a ball of radius $R\ge\Delta$, will contain at most 
$\left(\frac{2R}{\Delta}\right)^{O(\ddim)}$ net points.
% (as every ball of radius $\frac\Delta2$ contains at most a single net point, so the number of net points is bounded by the number of ).

\paragraph{Padded Decompositions and Sparse Covers.}
Consider a \emph{partition} $\mathcal{P}$ of $V$ into disjoint clusters.
For $v\in V$, we denote by $P(v)$ the cluster $P\in \mathcal{P}$ that contains $v$.
A partition $\mathcal{P}$ is strongly $\Delta$-\emph{bounded} (resp. weakly $\Delta$-bounded ) if the strong-diameter (resp. weak-diameter) of every $P\in\mathcal{P}$ is bounded by $\Delta$.
If the ball $B_G(v,\gamma\Delta)$ of radius $\gamma\Delta$ around a vertex $v$ is fully contained in $P(v)$, we say that $v$ is $\gamma$-{\em padded} by $\mathcal{P}$. Otherwise, if $B_G(v,\gamma\Delta)\not\subseteq P(v)$, we say that the ball is \emph{cut} by the partition.
\begin{definition}[Padded Decomposition]\label{def:PadDecompostion}
	A distribution $\mathcal{D}$ over partitions of a graph $G=\left(V,E,w\right)$ is strong (resp. weak) $(\beta,\delta,\Delta)$-padded decomposition if every $\mathcal{P}\in\supp(\mathcal{D})$ is strongly (resp. weakly) $\Delta$-bounded and for any $0\le\gamma\le\delta$, and $z\in V$,
	$$\Pr[B_G(z,\gamma\Delta)\subseteq P(z)] \ge e^{-\beta\gamma}~.$$
	\emph{(Probabilistic Decomposition)} A distribution $\mathcal{D}$ over partitions of a graph $G=\left(V,E,w\right)$ is strong (resp. weak) $(\beta,\Delta)$-probabilistic decomposition if every $\mathcal{P}\in\supp(\mathcal{D})$ is strongly (resp. weakly) $\Delta$-bounded and for every pair  $u,v\in V$,
	$$\Pr[P(v)\ne P(u)] \le \beta\cdot \frac{d_G(u,v)}{\Delta}~.$$
	\emph{(ThProbabilistic Decomposition)} A distribution $\mathcal{D}$ over partitions of a graph $G=\left(V,E,w\right)$ is strong (resp. weak) $(\beta,p,\Delta)$-threshold probabilistic (abr. ThProbabilistic) decomposition if every $\mathcal{P}\in\supp(\mathcal{D})$ is strongly (resp. weakly) $\Delta$-bounded and for every pair  $u,v\in V$, such that $d_G(u,v)\le \frac{\Delta}{\beta}$,
	$$\Pr[P(v)= P(u)] \ge p~.$$
	
	We say that a graph $G$ admits a strong (resp. weak) $(\beta,\delta)$-padded decomposition scheme 
	$\big/$ $\beta$-probabilistic decomposition scheme $\big/$ $(\beta,p)$-ThProbabilistic decomposition scheme, if for every parameter $\Delta>0$ it admits a strong (resp. weak) $(\beta,\delta,\Delta)$-padded decomposition  
	$\big/$ $(\beta,\Delta)$-probabilistic decomposition $\big/$ $(\beta,p,\Delta)$-ThProbabilistic decomposition that can be sampled in polynomial time.  
\end{definition}

We observe that padded decompositions are stronger than probabilistic decompositions, which by themselves stronger than ThProbabilistic decompositions. 

\begin{observation}\label{obs:PadToProb}
	Suppose that a weighted graph  $G=(V,E,w)$ admits a strong/weak $(\beta,\delta,\Delta)$-padded decomposition $\mathcal{D}$ such that $\delta\ge \frac 1\beta$. Then $\mathcal{D}$ is also a strong/weak $(\beta,\Delta)$-probabilistic decomposition.
\end{observation}
\begin{proof}
	Let $v,u\in V$ be a pair of vertices. If $d_G(u,v)\ge \frac\Delta\beta$, then obviously $\Pr[P(v)\ne P(u)] \le 1\le \beta\cdot \frac{d_G(u,v)}{\Delta}$.
	Thus we can assume $d_G(u,v)\le \frac\Delta\beta\le \delta\Delta$. Set $\gamma=\frac{d_G(u,v)}{\Delta}$. It holds that 
	\[
	\Pr[P(v)=P(u)]\ge\Pr\left[B_{G}\left(v,\gamma\Delta\right)\subseteq P(v)\right]\ge e^{-\beta\gamma}\ge 1-\beta\gamma~.
	\]
	In particular, $\Pr[P(v)\ne P(u)]\le\beta\gamma=\beta\cdot\frac{d_G(u,v)}{\Delta}$ as required.
\end{proof}
\begin{observation}\label{obs:ProbtoThProb}
	Suppose that a weighted graph $G=(V,E,w)$ admits a strong/weak
	$(\beta,\Delta)$-probabilistic decomposition $\mathcal{D}$.
	Then for every $\beta\le\beta'$, $\mathcal{D}$ is also a strong/weak $(\beta',\frac{\beta}{\beta'},\Delta)$-probabilistic decomposition.
\end{observation}
\begin{proof}
	Let $v,u\in V$ be a pair of vertices at distance at most $d_G(u,v)\le \frac{\Delta}{\beta'}$, then $\Pr[P(v)\ne P(u)]\le\frac{\beta}{\Delta}\cdot\frac{\Delta}{\beta'}=\frac{\beta}{\beta'}$, and hence $\Pr[P(v)= P(u)]\ge1-\frac{\beta}{\beta'}$.
\end{proof}

%\atodoin{Stopped here, moved def to here}
%\begin{definition}[Probabilistic Decomposition]\label{def:separating}
%	A distribution $\mathcal{D}$ over partitions of a graph $G=\left(V,E,w\right)$ is strong (resp. weak) $(\beta,\Delta)$-probabilistic decomposition if every $\mathcal{P}\in\supp(\mathcal{D})$ is strongly (resp. weakly) $\Delta$-bounded and for every pair  $u,v\in V$,
%	$\Pr[P(v)\ne P(u)] \le \beta\cdot \frac{d_G(u,v)}{\Delta}$.
%	We say that a graph $G$ admits a strong (resp. weak) $\beta$-probabilistic decomposition scheme, if for every parameter $\Delta>0$ it admits a strong (resp. weak) $(\beta,\Delta)$-probabilistic decomposition that can be sampled in polynomial time.  \\
%	\emph{(Probabilistic Decomposition)}
%\end{definition} 
A related notion to padded decompositions is sparse covers.
\begin{definition}[Sparse Cover]
	A collection of clusters $\mathcal{C} = \{C_1,..., C_t\}$ is called a weak/strong $(\beta,s,\Delta)$ sparse cover if the following conditions hold.
	\begin{enumerate}
		\item Bounded diameter: The weak/strong diameter of every $C_i\in\mathcal{C}$ is bounded by $\Delta$.\label{condition:RadiusBlowUp}
		\item Padding: For each $v\in V$, there exists a cluster $C_i\in\mathcal{C}$ such that $B_G(v,\frac\Delta\beta)\subseteq C_i$.
		\item Overlap: For each $v\in V$, there are at most $s$ clusters in $\mathcal{C}$ containing $v$.		
	\end{enumerate}
	We say that a graph $G$ admits a weak/strong $(\beta,s)$ sparse cover scheme, if for every parameter $\Delta>0$ it admits a weak/strong $(\beta,s,\Delta)$ sparse cover that can be constructed in expected polynomial time. 
\end{definition}

\paragraph{Truncated Exponential Distributions.}To create padded decompositions, similarly to previous works, we will use truncated exponential distribution. That is, exponential distribution conditioned on the event that the outcome lays in a certain interval. 
The \emph{$[\theta_1,\theta_2]$-truncated exponential distribution} with
parameter $\lambda$ is denoted by $\Texp_{[\theta_1, \theta_2]}(\lambda)$, and
the density function is:
$f(y)= \frac{ \lambda\, e^{-\lambda\cdot y} }{e^{-\lambda \cdot \theta_1} - e^{-\lambda
		\cdot \theta_2}}$, for $y \in [\theta_1, \theta_2]$.
For the \emph{$[0,1]$-truncated exponential distribution} we drop the
subscripts and denote it by $\Texp(\lambda)$, the density function is then
$f(y) = \frac{ \lambda\cdot e^{-\lambda\cdot y} }{1 - e^{-\lambda}}$.

\section{Strongly Padded Decomposition}\label{sec:centersToPadded}
In this section we prove the main technical theorem of this paper.  
\begin{theorem}\label{thm:padded}
	Let $G=(V,E,w)$ be a weighted graph and $\Delta>0,\tau=\Omega(1)$ parameters. Suppose that we are given a set $N\subseteq V$ of center vertices such that for every $v\in V$:
	\begin{itemize}
		\item \textsc{Covering.} There is $x\in N$ such that $d_G(v,x)\le \Delta$.
		\item \textsc{Packing.} There are at most $\tau$ vertices in $N$ at distance $3\Delta$, i.e. $\left|B_G(v,3\Delta)\cap N\right|\le \tau$.
%		\item \red{\textsc{Packing.} There are at most $\tau$ vertices in $N$ at distance $(1+2\alpha)\Delta$, i.e. $\left|B_G(v,(1+2\alpha)\Delta)\cap N\right|\le \tau$.}
	\end{itemize}
	Then $G$ admits a strong $\left(O(\ln \tau),\frac{1}{16},4\Delta\right)$-padded decomposition that can be efficiently sampled.
\end{theorem}
We start with description of the \cite{MPX13} algorithm (with some adaptations), and its properties. Later, in \Cref{subsec:ThmPadProof} we will prove \Cref{thm:padded}. 
%An alternative construction is given in \Cref{sec:cones}.

\subsection{Clustering Algorithm Based on Starting Times}\label{subsec:clustering}
As we make some small adaptations, and the role of the clustering algorithm is essential, we provide full details.
Let $\Delta>0$ be some parameter and let $N\subseteq V$ be some set of centers such that for every $v\in V$, $d_G(v,N)\le \Delta$. 
For each center $x\in N$, let $\delta_x\in [0,\Delta]$ be some parameter. The choice of $\{\delta_x\}_{x\in N}$ differs among different implementations of the algorithm. In our case we will sample $\delta_x$ using truncated exponential distribution.
Each vertex $v$ will join the cluster $C_x$ of the center $x\in N$ for which the value $\delta_x-d_G(x,v)$ is maximized. Ties are broken in a consistent manner, that is we have some order $x_1,x_2,\dots$ . Among the centers $x_i$ that minimize $\delta_{x_i}-d_G(x_i,v)$, $v$ will join the cluster of the center with minimal index.
Note that it is possible that a center $x\in N$ will join the cluster of a different center $x'\in N$.
An intuitive way to think about the clustering process is as follows: each center $x$ wakes up at time $-\delta_x$ and begins to ``spread'' in a continuous manner. The spread of all centers done in the same unit tempo. A vertex $v$ joins the cluster of the first center that reaches it. 

\begin{claim}\label{claim:StrongDiam}
	Every non-empty cluster $C_x$ created by the algorithm has strong diameter at most $4\Delta$ .
\end{claim}
\begin{wrapfigure}{r}{0.13\textwidth}
	\begin{center}
		\vspace{-20pt}
		\includegraphics[width=0.13\textwidth]{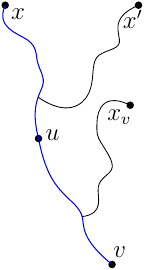}
		\vspace{-17pt}
	\end{center}
	\vspace{-10pt}
\end{wrapfigure}

%\begin{proof}
\emph{Proof.}
	Consider a vertex $v\in C_x$. We argue that $d_G(v,x)\le 2\Delta$. This will already imply that $C_x$ has weak diameter $4\Delta$. 
	Let $x_v$ be the closest center to $v$, then $d_G(v,x_v)\le \Delta$. As $v$ joined the cluster of $x$, it holds that
	$\delta_{x}-d_{G}(v,x)\ge\delta_{x_v}-d_{G}(v,x_v)$. In particular $d_{G}(v,x)\le\delta_{x}+d_{G}(v,x_v)\le2\Delta$. 
	
	Let $\mathcal{I}$ be the shortest path in $G$ from $v$ to $x$ (the blue path on the illustration on the right). 
	For every vertex $u\in\mathcal{I}$ and center $x'\in N$, it holds that
	\[
	\delta(x)-d_{G}(u,x)=\delta(x)-\left(d_{G}(v,x)-d_{G}(v,u)\right)\stackrel{(*)}{\ge}\delta(x')-d_{G}(v,x')+d_{G}(v,u)\ge\delta(x')-d_{G}(u,x')~.
	\]
	If the inequality $^{(*)}$ is strict, then $\delta(x)-d_{G}(u,x)>\delta(x')-d_{G}(u,x')$ and $u$ will prefer $x$ over $x'$. Otherwise, necessarily $\delta(x)-d_{G}(v,x)=\delta(x')-d_{G}(v,x')$ and $x$ has smaller index then $x'$ (as it was preferred by $v$). In particular
 	$\delta(x)-d_{G}(u,x)=\delta(x')-d_{G}(u,x')$ and hence $u$ will prefer $x$ over $x'$. 	
 	We conclude that $u$ will prefer the center $x$ over any other center. It follows that $\mathcal{I}\subseteq C_x$. In particular, $d_{G[C_x]}(v,x)\le2\Delta$. 
 	The claim now follows.
 	\qed
%\end{proof}

\begin{claim}\label{claim:PadProperty}
	Consider a vertex $v$, and let $x_{(1)},x_{(2)},\dots$ be an ordering of the centers w.r.t. $\delta(x_{(i)})-d_G(v,x_{(i)})$. That is $\delta(x_{(1)})-d_G(v,x_{(1)})\ge \delta(x_{(2)})-d_G(v,x_{(2)})\ge\dots$ .  
	Set $\Upsilon=(\delta(x_{(1)})-d_G(v,x_{(1)})) - (\delta(x_{(2)})-d_G(v,x_{(2)}))$. Then for every vertex $u$ such that $d_G(v,u)<\frac{\Upsilon}{2}$ it holds that $u\in C_{x_{(1)}}$.
\end{claim}
\begin{proof}
	For every center $x_{(i)}\ne x_{(1)}$ it holds that,
	\begin{align*}
	\delta(x_{1})-d_{G}(u,x_{(1)}) & >\delta(x_{(1)})-d_{G}(v,x_{(1)})-\frac{\Upsilon}{2}\ge\delta(x_{(i)})-d_{G}(v,x_{(i)})+\frac{\Upsilon}{2}>\delta(x_{(i)})-d_{G}(u,x_{(i)})~.
	\end{align*}
	In particular, $u\in C_{x_1}$. 
\end{proof}

\subsection{Proof of \Cref{thm:padded}}\label{subsec:ThmPadProof}
For every center $x\in N$, we sample $\delta'_x\in[0,1]$ according to $\Texp(\lambda)$ truncated exponential distribution with parameter $\lambda=2+2\ln \tau$. Set $\delta_x=\delta'_x\cdot\Delta\in[0,\Delta]$. We execute the clustering algorithm from \Cref{subsec:clustering} with parameters $\{\delta_x\}_{x\in N}$ to get a partition $\mathcal{P}$.

According to \Cref{claim:StrongDiam}, we created a distribution over strongly $4\Delta$-bounded partitions. 
Consider some vertex $v\in V$ and parameter $\gamma\le\frac14$. We will argue that the ball $B=B_G(v,\gamma \Delta)$ is fully contained in $P(v)$ with probability at least $e^{-O(\gamma\log\tau)}$.
Let $N_v$ be the set of centers $x$ for which there is non zero probability that $C_{x}$ intersects $B$. 
Following the calculation in \Cref{claim:StrongDiam}, each vertex joins the cluster of a center at distance at most $2\Delta$.
By triangle inequality, all the centers in $N_v$ are at distance at most $(2+\gamma)\Delta\le 3\Delta$ from $v$. In particular $|N_v|\le \tau$.

Set $N_v=\{x_1,x_2,\dots\}$ ordered arbitrarily. 
Denote by $\mathcal{F}_i$ the event that $v$ joins the cluster of $x_i$, i.e. $v\in C_{x_i}$. 
Denote by $\mathcal{C}_i$ the event that $v$ joins the cluster of $x_i$, but not all of the vertices in $B$ joined that cluster, that is $v\in C_{x_i}\cap B\neq B$. 
To prove the theorem, it is enough to show that $\Pr\left[\cup_{i}\mathcal{C}_{i}\right]\le 1-e^{-O(\gamma\cdot\lambda)}$.
Set $\alpha=e^{-2\gamma\cdot\lambda}$.

\begin{claim}
	For every $i$, 
	$\Pr\left[\mathcal{C}_{i}\right]\le\left(1-\alpha\right)\left(\Pr\left[\mathcal{F}_{i}\right]+\frac{1}{e^{\lambda}-1}\right)$.
\end{claim}
\begin{proof}
	As the order in $N_v$ is arbitrary, assume w.l.o.g. that $i=|N_v|$ and denote $x=x_{|N_v|}$, $\mathcal{C}=\mathcal{C}_i$, $\mathcal{F}=\mathcal{F}_i$, $\delta=\delta_{x_i}$ and $\delta'=\delta'_{x_i}$.
	Let $X\in [0,1]^{|N_v|-1}$ be the vector where the $j$'th coordinate equals $\delta'_{x_j}$. Set  $\rho_{X}=\frac{1}{\Delta}\cdot\left(d_{G}(x,v)+\max_{j<|N_v|}\left\{ \delta_{x_{j}}-d_{G}(x_{j},v)\right\} \right)$. Note that $\rho_X$ is the minimal value of $\delta'$ such that if $\delta'>\rho_X$, then $x$ has the maximal value $\delta_{x}-d_{G}(x,v)$, and therefore $v$ will join the cluster of $x$. 
	Note that it is possible that $\rho_X>1$.
	Conditioning on the samples having values $X$, and assuming that $\rho_{X}\le 1$ it holds that
	\[
	\Pr\left[\mathcal{F}\mid X\right]=\Pr\left[\delta'>\rho_{X}\right]=\int_{\rho_{X}}^{1}\frac{\lambda\cdot e^{-\lambda y}}{1-e^{-\lambda}}dy=\frac{e^{-\rho_{X}\cdot\lambda}-e^{-\lambda}}{1-e^{-\lambda}}~.
	\]
	If $\delta'>\rho_{X}+2\gamma$ then $\delta-d_{G}(x,v)>\max_{j\ne i}\left\{ \delta_{x_{i}}-d_{G}(x_{i},v)\right\} +2\gamma\Delta$. In particular, by \Cref{claim:PadProperty} the ball $B$ will be contained in $C_{x}$.
	We conclude
	\begin{align*}
	\Pr\left[\mathcal{C}\mid X\right] & \le\Pr\left[\rho_{X}\le\delta'\le\rho_{X}+2\gamma\right]\\
	& =\int_{\rho_{X}}^{\max\left\{ 1,\rho_{X}+2\gamma\right\} }\frac{\lambda\cdot e^{-\lambda y}}{1-e^{-\lambda}}dy\\
	& \le\frac{e^{-\rho_{X}\cdot\lambda}-e^{-\left(\rho_{X}+2\gamma\right)\cdot\lambda}}{1-e^{-\lambda}}\\
	& =\left(1-e^{-2\gamma\cdot\lambda}\right)\cdot\frac{e^{-\rho_{X}\cdot\lambda}}{1-e^{-\lambda}}\\
	& =\left(1-\alpha\right)\cdot\left(\Pr\left[\mathcal{F}\mid X\right]+\frac{1}{e^{\lambda}-1}\right)~.
	\end{align*}
	Note that if $\rho_{X}> 1$ then $\Pr\left[\mathcal{C}\mid X\right]=0\le \left(1-\alpha\right)\cdot\left(\Pr\left[\mathcal{F}\mid X\right]+\frac{1}{e^{\lambda}-1}\right)$ as well.
	Denote by $f$ the density function of the distribution over all possible values of $X$. Using the law of total probability, we can bound the probability that the cluster of $x$ cuts $B$
	\begin{align*}
	\Pr\left[\mathcal{C}\right] & =\int_{X}\Pr\left[\mathcal{C}\mid X\right]\cdot f(X)~dX\\
	& \le\left(1-\alpha\right)\cdot\int_{X}\left(\Pr\left[\mathcal{F}\mid X\right]+\frac{1}{e^{\lambda}-1}\right)\cdot f(X)~dX\\
	& =\left(1-\alpha\right)\cdot\left(\Pr\left[\mathcal{F}\right]+\frac{1}{e^{\lambda}-1}\right)
	\end{align*}
\end{proof}

We bound the probability that the ball $B$ is cut. 
\begin{align*}
\Pr\left[\cup_{i}\mathcal{C}_{i}\right]=\sum_{i=1}^{|N_{v}|}\Pr\left[\mathcal{C}_{i}\right] & \le\left(1-\alpha\right)\cdot\sum_{i=1}^{|N_{v}|}\left(\Pr\left[\mathcal{F}_{i}\right]+\frac{1}{e^{\lambda}-1}\right)\\
& \le\left(1-e^{-2\gamma\cdot\lambda}\right)\cdot\left(1+\frac{\tau}{e^{\lambda}-1}\right)\\
& \le\left(1-e^{-2\gamma\cdot\lambda}\right)\cdot\left(1+e^{-2\gamma\cdot\lambda}\right)=1-e^{-4\gamma\cdot\lambda}~,
\end{align*}
where the last inequality follows as
$e^{-2\gamma\lambda}=\frac{e^{-2\gamma\lambda}\left(e^{\lambda}-1\right)}{e^{\lambda}-1}\ge\frac{e^{-2\gamma\lambda}\cdot e^{\lambda-1}}{e^{\lambda}-1}\ge\frac{e^{\frac{\lambda}{2}-1}}{e^{\lambda}-1}=\frac{\tau}{e^{\lambda}-1}$.
To conclude, we obtain a strongly $4\Delta$-bounded partition, such that for every $\gamma\le\frac{1}{16}$ and $v\in V$, the ball $B_G(v,\gamma\cdot4\Delta)$ is fully contained in a single cluster with probability at least \[
\Pr\left[B(v,\gamma\cdot4\Delta)\subseteq P(v)\right]\ge e^{-4\cdot(4\gamma)\cdot\lambda}=e^{-\gamma\cdot32(1+\ln\tau)}~.
\]
\qed

%\begin{remark}
%	One can generalize \Cref{thm:padded}: suppose that for some parameters $\delta,\alpha$, where $0<\delta<\frac\alpha2<\frac12$, there is a set $N$ of centers such that for every $v\in V$, $|B_H(v,(1+\alpha+\delta)\Delta)\cap N|\le\tau$. Then $G$ admits a strong $\left(\frac{4}{\alpha\cdot(1-\frac{2\delta}{\alpha})}\cdot\ln(e\tau),\delta,2(1+\alpha+\delta)\Delta\right)$-padded decomposition.
%%	
%	To obtain this result one should follow the same construction as above, and sample the shifts using truncated exponential distribution in $[0,1]$ with parameter $\lambda=\frac{\ln e\tau}{(1-\frac{2\delta}{\alpha})}$.
%\end{remark}
%\begin{remark}
%	Actually we can prove a generalized version of \Cref{thm:padded}. Suppose that there is a set $N$ of centers such that each vertex $v\in V$ has at least one center at distance at most $\Delta$ and at most $\tau_v$ centers at distance $3\Delta$. Then for every parameter $\lambda=\Omega(1)$, there is a distribution over partitions with strong diameter $4\Delta$ such that for every parameter $\gamma\in(0,\frac14)$, the ball around every vertex $v$ of radius $\gamma\Delta$ is cut with probability at most $(1-e^{-2\gamma\lambda})(1+\frac{\tau_v}{e^\lambda-1})$.
%\end{remark}

\section{Doubling Dimension}\label{sec:doubling}
Our strongly padded decompositions for doubling graphs are a simple corollary of \Cref{thm:padded}.
\begin{corollary}\label{cor:PaddedDoubling}
	Let $G=(V,E,w)$ be a weighted graph with doubling dimension $\ddim$. Then $G$ admits a strong $\left(O(\ddim),\Omega(1)\right)$-padded decomposition scheme.
\end{corollary}
\begin{proof}
	Fix some $\Delta>0$. Let $N$ be a $\Delta$-net of $X$. According to \Cref{lem:doubling_packing}, for every vertex $v$, the number of net points at distance $3\Delta$ is bounded by $2^{O(\ddim)}$. The corollary follows by \Cref{thm:padded}.
\end{proof}

Next, we construct a sparse cover scheme.
\begin{theorem}\label{thm:DdimCover}
	Let $G=(V,E,w)$ be a weighted graph with doubling dimension $\ddim$ and parameter $t\ge1$. Then $G$ admits a strong $\left(O(t),O(2^{\nicefrac{\ddim}{t}}\cdot\ddim\cdot\log t)\right)$ sparse cover scheme. 
	In particular, there is a strong $\left(\ddim,O(\ddim\cdot\log \ddim)\right)$ sparse cover scheme.	
	
\end{theorem}
\begin{proof}
	Let $\Delta>0$ be the diameter parameter. Let $\alpha=\theta(1)$ be a constant to be determined later, set $\beta=\alpha\cdot t$. We will construct a strong $\left(\beta, O(2^{\nicefrac{\ddim}{t}}\cdot\ddim\cdot\log t),4\Delta\right)$ sparse cover. As $\Delta$ is arbitrary, this will imply strong $\left(4\beta, O(2^{\nicefrac{\ddim}{t}}\cdot\ddim\cdot\log t)\right)$ sparse cover scheme.
	
	The sparse cover is constructed by sampling $O(2^{\nicefrac{\ddim}{t}}\cdot\ddim\cdot\log t)$ independent partitions using  \Cref{cor:PaddedDoubling} with diameter parameter $\Delta$, and taking all the clusters from all the partitions to the cover. The sparsity and strong diameter properties are straightforward. To argue that each vertex is padded in some cluster we will use the constructive version of the Lov\'asz Local Lemma by Moser and Tardos \cite{MT10}.
	\begin{lemma}[Constructive Lov\'asz Local Lemma \cite{MT10}] \label{lem:lovasz}
		Let $\mathcal{P}$ be a finite set of mutually independent random variables in a probability space. Let $\mathcal{A}$ be a set of events determined by these variables. For $A\in\mathcal{A}$ let $\Gamma(A)$ be a subset of $\mathcal{A}$ satisfying that $A$ is independent from the collection of events $\mathcal{A}\setminus(\{A\}\cup\Gamma(A))$.
		If there exist an assignment of reals $x:\mathcal{A}\rightarrow(0,1)$  such that
		\[
		\forall A\in \mathcal{A}~:~~\Pr[A]\le x(A)\cdot\Pi_{B\in\Gamma(A)}(1-x(B))~,
		\]
		then there exists an assignment to the variables $\mathcal{P}$ not violating any of the events in $\mathcal{A}$. Moreover, there is an algorithm that finds such an assignment in expected time  $\sum_{A\in\mathcal{A}}\frac{x(A)}{1-x(A)}\cdot\poly \left(|\mathcal{A}|+|\mathcal{P}|\right)$.
	\end{lemma} 
	
	Formally, recall the construction of \Cref{thm:padded} used in \Cref{cor:PaddedDoubling}.
	Let $N$ be a $\Delta$-net, that we will use as centers. 
	Consider an arbitrary vertex $v\in V$, and fix some sample of the starting times $\{\delta_x\}_{x\in N}$. Let $x_{(1)}$ be the vertex maximizing $\delta_x-d_G(x,v)$ and $x_{(2)}$ the second largest. In other words, $\delta_{x_{(1)}}-d_G(x_{(1)},v)\ge \delta_{x_{(2)}}-d_G(x_{(2)},v)\ge\max_{x\in N\setminus\{x_{(1)},x_{(2)}\}}\{\delta_{x}-d_G(x,v)\}$.
	Let $\Psi_v$ be the event that $(\delta_{x_{(1)}}-d_G(x_{(1)},v))- (\delta_{x_{(2)}}-d_G(x_{(2)},v))<4\frac\Delta\beta$. Recall that the event that the ball of radius $2\frac\Delta\beta$ around $v$ is cut contained in $\Psi_v$. Following the analysis of \Cref{thm:padded}, 
	$\Pr\left[\Psi_{v}\right]\le1-e^{-O(\ddim\cdot\frac{1}{\beta})}=1-2^{-\ddim/t}$, where the equality holds for an appropriate choice of $\alpha$.
	
	Let $x_v$ be the closest center to $v$. It holds that $\delta_{x_v}-d_G(x_v,v)\ge-\Delta$, while for every center $x$ at distance larger that $3\Delta$ it holds that $\delta_{x}-d_G(x,v)\le-2\Delta$. Therefore $\Psi_{v}$ depends only on centers at distance at most $3\Delta$. In particular, by triangle inequality, if $v$ and $u$ are farther away than $6\Delta$, $\Psi_{v}$ and $\Psi_{u}$ are independent.	
	
	We take $m=\alpha_m\cdot 2^{\frac{\ddim}{t}}\cdot\ddim\cdot\log t$ independent partitions of $X$ using  \Cref{cor:PaddedDoubling}, for $\alpha_m=\Theta(1)$ to be determined later. Denote by $\Psi^i_v$ the event representing $\Psi_v$ in the $i$'th partition. 
	Let $\Phi_{v}=\bigwedge_{i=1}^{m}\Psi_{v}^{i}$ be the event that $v$ ``failed'' in all the partitions. It holds that 
	\[
	\Pr[\Phi_{v}]\le\left(1-2^{-\ddim/t}\right)^{m}\le e^{-2^{-\ddim/t}\cdot m}=e^{-\alpha_{m}\cdot\ddim\cdot\log t}~.
	\]
	Note that if $\Phi_v$ did not occurred, then the ball of radius $2\frac\Delta\beta$ around $v$ is contained in a single cluster in at least one partition.
	
	Let $Y$ be an $\frac\Delta\beta$-net of $X$. 
	Set $\mathcal{A}=\{\Phi_v\}_{v\in Y}$, to be a set of events determined by $\{\delta^i_x\}_{x\in N,1\le i\le m}$ ($\delta^i_x$ denotes $\delta_x$ in the $i$'th partition).
	Each event $\Phi_v$ might depend only on events $\Phi_u$ corresponding to vertices $u$ at distance at most $6\Delta$ from $v$. By \Cref{lem:doubling_packing}, $\Phi_v$ is independent of all, but $|\Gamma(\Phi_v)|\le\left(\frac{12\Delta}{\Delta/\beta}\right)^{O(\ddim)}=2^{O(\ddim\cdot\log t)}$ events.
	For every $\Phi_v\in\mathcal{A}$, set $x(\Phi_v)=p= 2^{-O(\ddim\cdot\log t)}$, such that $\max_{v\in Y}|\Gamma(v)|\le\frac{1}{2p}$. Then, for every $\Phi_v\in\mathcal{A}$ it holds that,
	\[
	x(\Phi_{v})\cdot\Pi_{B\in\Gamma(\Phi_{v})}(1-x(B))=p\cdot\left(1-p\right)^{|\Gamma(\Phi_{v})|}\ge p\cdot\left(1-p\right)^{\frac{1}{2p}}\ge\frac{p}{e}
	\ge\Pr(\Phi_{v})~,
	\]
	where the last inequality holds for large enough $\alpha_m$. 
	By \Cref{lem:lovasz} we can find an assignment to $\{\delta^i_x\}_{x\in N,1\le i\le m}$ such that none of the events $\{\Phi_v\}_{v\in Y}$ occurred in $\sum_{A\in\mathcal{A}}\frac{x(A)}{1-x(A)}\cdot\poly\left(|\mathcal{A}|+|\mathcal{P}|\right)=\frac{p}{1-p}\cdot\poly\left(n\right)=\poly\left(n\right)$ expected time.
	Under this assumption, we argue that our sparse cover has the padding property.
	Consider some vertex $v\in V$. There is a net point $u\in Y$ at distance at most $\frac\Delta\beta$ from $v$. As the event $\Phi_u$ did not occur, there is some cluster $C$ in the cover in which $u$ is padded. In particular 
	$B_G(v,\gamma\Delta)\subseteq B_G(u,2\gamma\Delta)\subseteq C$ as required.
	
%	Suppose that $|V|=n$, then the running time is 
%	$|Y|\cdot\frac{p}{1-p}\cdot\poly\left(|Y|+|Y|\right)=\poly(n)$.
\end{proof}

\section{Minor Free Graphs}\label{sec:MinorFree}
Our clustering algorithm is based on the clustering algorithm of \cite{AGGNT19}, with a small modification.
The clustering of \cite{AGGNT19} has two steps. In the first step the graph is partitioned into $r$-\emph{Core} clusters (see \Cref{def:rCore} below). While $r$-core clusters do not have bounded diameter, they do have a simple geometric structure. Moreover, this clustering also has the padding property for small balls. In the second step, each $r$-core cluster is partitioned into bounded diameter sub-clusters using \Cref{thm:padded}.

\begin{definition}[$r$-Core]\label{def:rCore}
	Given a weighted graph $G=(V,E,w)$, 
	an $r$-\emph{core} with radius $\Delta$ is a set of at most $r$ shortest paths $\mathcal{I}_1,\dots,\mathcal{I}_{r'}$ such that for every $v\in V$, $d_G(v,\cup_i \mathcal{I}_i)\le\Delta$.
	
	Given a cluster $C\subseteq G$, we say that $C$ is an $r$-core cluster with radius $\Delta$, if $G[C]$ has an $r$-\emph{core} with radius $\Delta$. 
	Given a partition $\mathcal{P}$ of $G$, we say that it is an  $r$-core partition with radius $\Delta$ if each cluster $C\in \mathcal{P}$, is an $r$-core cluster with radius $\Delta$.
\end{definition}

The following theorem was proved implicitly in \cite{AGGNT19}.
\begin{lemma}[Core Clustering \cite{AGGNT19}]\label{lem:coreClustering}
	Given a weighted graph $G=(V,E,w)$ that excludes $K_r$ as a minor and a parameter $\Delta>0$, there is a distribution $\mathcal{D}$ over $r$-core partitions with radius $\Delta$, such that for every vertex $v\in V$ and $\gamma\in(0,\Omega(\frac1r))$ it holds that
	\[
	\Pr
	\left[B_{G}(v,\gamma\Delta)\subseteq P(v)\right]\ge e^{-O(r\cdot\gamma)}~.
	\]
\end{lemma}
We provide a proof of \Cref{lem:coreClustering} that differs from \cite{AGGNT19}, and arguably simplifies it. See  \Cref{subsec:coreClustering}, and the discussion therein.
We now proceed to proving our main \Cref{thm:StrongMinorFree}.
%Even though we will not provide full details of the proof of \Cref{lem:coreClustering}, we will describe the algorithm itself and provide some intuition for the core clustering in \Cref{subsec:coreClustering}.
Our clustering algorithm will be executed in two steps: first we partition the graph into $r$-core clusters (\Cref{lem:coreClustering}) and then we partition each $r$-core cluster using \Cref{thm:padded}.

\begin{hisnote*}
%Historical notes: 
\cite{AGGNT19} presented two different algorithms for strong and weak padded decompositions. Each of these algorithms consist	ed of two steps. 
For weak decompositions, essentially they first partitioning the graph into $r$-core clusters. Secondly, instead of partition further each cluster, they pick a net from the $r$-cores in all the clusters, and iteratively grow balls around net points, ending with weak diameter guarantee.
For strong decompositions, they partition the graph into $1$-core clusters (instead of $r$-core), ending with a probability of only $ e^{-O(r^2\cdot\gamma)}$ for a vertex $x$ to be $\gamma$-padded.
\end{hisnote*}
\subsection{Strong Padded Partitions for $K_r$ Minor Free Graphs}\label{subsec:MinorFreeClustering}

\begin{lemma}\label{lem:coreToPadded}
	Let $G=(V,E,w)$ be a weighted graph that has an $r$-core with radius $\Delta$. Then $G$ admits a strong $\left(O(\log r),\frac{1}{16},8\Delta\right)$-padded decomposition.
\end{lemma}
%\begin{proof}
%	Let $\mathcal{I}_1,\mathcal{I}_2,\dots,\mathcal{I}_{r'}$ be the $r$-core of $G$. For each $i$, let $N_i$ be a $\frac\Delta8$-net of $\mathcal{I}_i$. Set $N=\cup_i N_i$.
%	%
%	Every vertex $v\in V$ has some vertex in $N$ at distance at most $\frac\Delta4$. Indeed, by definition of $r$-core, there is $x\in \mathcal{I}_i$ such that $d_{G}(v,x)\le\frac\Delta8$. Furthermore, there is a net point $y\in N_i$ at distance at most $\frac\Delta8$ from $x$. By triangle inequality $d_{G}(v,y)\le\frac\Delta4$.
%	As $\mathcal{I}_i$ is a shortest path and $N_i$ is a $\frac\Delta8$-net, there are at most $O(1)$ net points at distance $\frac34\Delta$ from $v$ in $N_i$. We conclude that in $N$ there are at most $O(r)$ net points at distance $\frac34\Delta$ from $v$. The lemma now follows by \Cref{thm:padded}.
%\end{proof}
\begin{proof}
	Let $\mathcal{I}_1,\mathcal{I}_2,\dots,\mathcal{I}_{r'}$ be the $r$-core of $G$. For each $i$, let $N_i$ be a $\Delta$-net of $\mathcal{I}_i$. Set $N=\cup_i N_i$.
	Every vertex $v\in V$ has some vertex in $N$ at distance at most $2\cdot\Delta$. Indeed, by definition of $r$-core, there is $x\in \mathcal{I}_i$ such that $d_{G}(v,x)\le\Delta$. Furthermore, there is a net point $y\in N_i$ at distance at most $\Delta$ from $x$. By triangle inequality $d_{G}(v,y)\le2\cdot\Delta$.
	As $\mathcal{I}_i$ is a shortest path and $N_i$ is a $\Delta$-net, there are at most $O(1)$ net points at distance $6\Delta$ from $v$ in $N_i$. We conclude that in $N$ there are at most $O(r)$ net points at distance $6\Delta$ from $v$. The lemma now follows by \Cref{thm:padded}.
\end{proof}
\begin{theorem}\label{thm:StrongMinorFree}
	Let $G=(V,E,w)$ be a weighted graph that excludes $K_r$ as a minor. Then $G$ admits a strong $\left(O(r),\Omega(\frac1r)\right)$-padded decomposition scheme.
\end{theorem}
\begin{proof}
	Let $\Delta>0$ be some parameter.
	We construct the decomposition in two steps. First we sample an $r$-core partition $\mathcal{P}$ with radius parameter $\frac\Delta8$ using \Cref{lem:coreClustering}.
	Next, for every cluster $C\in \mathcal{P}$, we create a partition $\mathcal{P}_C$ using \Cref{lem:coreToPadded}. The final partition is simply $\cup_{C\in\mathcal{P}}\mathcal{P}_C$, the union of all the clusters in all the created partitions.
	It is straightforward that the created partition has strong diameter $\Delta$.
	To analyze the padding, consider a vertex $v\in V$ and parameter $0<\gamma \le \Omega(\frac1r)$.
	Denote by $C_v$ the cluster containing $v$ in $\mathcal{P}$, and by $P(v)$ the cluster of $v$ in the final partition. Then, 
	\begin{align*}
	\Pr\left[B_{G}(v,\gamma\Delta)\subseteq P(v)\right] & =\Pr\left[B_{G}(v,\gamma\Delta)\subseteq P(v)\mid B_{G}(v,\gamma\Delta)\subseteq C_{v}\right]\cdot\Pr\left[B_{G}(v,\gamma\Delta)\subseteq C_{v}\right]\\
	& \ge e^{-O(\gamma\cdot r)}\cdot e^{-O(\gamma\cdot\log r)}=e^{-O(\gamma\cdot r)}~,
	\end{align*}
	where we used the fact that conditioning on $B_{G}(v,\gamma\Delta)\subseteq C_{v}$, it holds that $B_{G}(v,\gamma\Delta)=B_{G[ C_{v}]}(v,\gamma\Delta)$.
\end{proof}

\subsection{The Core Clustering Algorithm: Proof of \Cref{lem:coreClustering}}\label{subsec:coreClustering}
%In this section we describe the construction of the partition from \Cref{lem:coreClustering}. 
In this section we prove \Cref{lem:coreClustering}, which was implicitly proved in \cite{AGGNT19}. Our algorithm and proof here follow similar lines to those in \cite{AGGNT19}. However, we made some small modifications that made the proof simpler and easier the follow (at least in the author's subjective opinion). 
Specifically, we pick the radii in \Cref{alg:coreClustering} using real exponential distribution, as opposed to truncated exponential distribution in \cite{AGGNT19}. This change makes the padding property to follow almost immediately, but more importantly the potential function argument is not as vague as in the original proof. 
The use of unbounded distribution however has the drawback that the resulting clusters might have radius (w.r.t. the $r$-core definition) larger than $\Delta$.
We deal with this issue by simply recursively running the algorithm on each such cluster. 
%Afterwards, we will provide some intuition regarding the proof. For full details, we refer to \cite{AGGNT19}.

Given two disjoint subsets $A,B\subseteq V$, we write $A \sim B$ if there exists an edge from a vertex in $A$ to some vertex in $B$. 
We denote the partition created by the algorithm by $\mathcal{S}$, and the clusters by $\{S_1,S_2,\dots\}$. 
The clusters are constructed iteratively. Initially $G_1=G$. At step $i$, $G_i=G\setminus\cup_{j=1}^{i-1} S_j$.
For a connected component
$C\in G_i$, let ${\cal K}_{|C}=\{S_j \mid j<i \wedge C\sim
S_j\}$ be the set of previously created clusters with a neighbor in $C_i$. 
To create $S_i$, pick arbitrary connected component $C_i$ in $G_i$, and a vertex $x_i\in C_i$. 
For every neighboring cluster $S_j\in {\cal K}_{|C_i}$, pick an arbitrary vertex $u_j\in C_i$ such that $u_j$ has a neighbor in $S_j$. 
For each such $u_j$, let $\mathcal{I}_j$ be an (arbitrary) shortest path in $G_i$ from $x_i$ to $u_j$. Let $T_i$ be the tree created by the union of $\{\mathcal{I}_j\}_{S_j\in {\cal K}_{|C_i}}$.\footnote{Note that there is always a way to pick $\{\mathcal{I}_j\}_{S_j\in {\cal K}_{|C_i}}$ such that $T_i$ will be a tree.}
Sample a radius parameter $R_i$ using exponential distribution $\exp(1)$ with parameter $1$ (the density function is $f(x)=e^{-x}$).
The cluster $S_i$ is defined as $B_{G_i}(T_i,R_i\Delta)$, the set of all vertices at distance at most $R_i\Delta$ from $T_i$ w.r.t. $d_{G_i}$. 
This finishes the construction of $S_i$. The algorithm halts when all the vertices are clustered.
The pseudo-code is presented in \Cref{alg:coreClustering}. See also \Cref{fig:Cops} for illustration of the algorithm.

\begin{algorithm}
	\caption{\texttt{Core-Partition}($G$,$\Delta$,$r$)}\label{alg:coreClustering}
	\begin{algorithmic}[1]
		\STATE Let $G_1 \leftarrow G$, $i\leftarrow 1$.
		\STATE Let $\mathcal{S} \leftarrow \emptyset$.
		\WHILE {$G_i$ is non-empty}
		\STATE Let $C_i$ be an arbitrary connected component of $G_i$.\label{line:PickCluster}
		\STATE Pick arbitrary $x_i \in C_i$. \label{alg:CenterChoose}
		For each $S_j\in {\cal K}_{|C_i}$, let $u_j\in C_i$ be some vertex with a neighbor in $S_j$.
		\STATE Let $T_i$ be a tree rooted at $x_i$ and consisting of shortest paths towards $\{u_j\mid S_j\in {\cal K}_{|C_i}\}$.
		\STATE Sample $R_i\sim\exp(1)$.
		\STATE Let $S_i \leftarrow B_{G_i}(T_i,R_i\Delta)$.
%		\IF {$R_i<?r$}
		\STATE Add $S_i$ to $\mathcal{S}$.
%		\ELSE
%		\STATE Let 
%		\ENDIF		
		\STATE $G_{i+1} \leftarrow G_i \setminus S_i$.
		\STATE $i\leftarrow i+1$.
		\ENDWHILE
		\RETURN $\mathcal{S}$.
	\end{algorithmic}
\end{algorithm}

\begin{figure}[h]
	\centering{\includegraphics[scale=0.8]{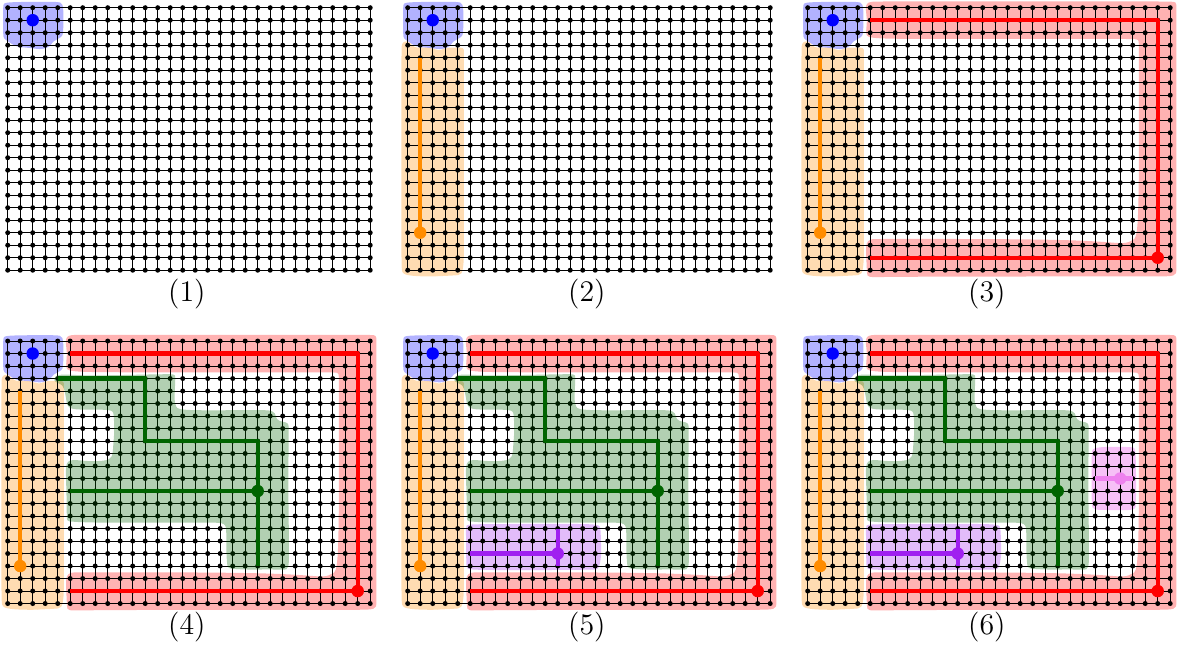}} 
	\caption{\label{fig:Cops}\small \it 
		The figure illustrates the $6$ first steps in \Cref{alg:coreClustering}. Here $G$ is the (weighted) grid graph. Note that $G$ excludes $K_5$ as a minor.
		In step $(4)$, $G_4$ is the graph induced by all the vertices not colored in blue, orange or red. $G_4$ has a single connected component $C_4$. The green vertex defined as $x_4$. ${\cal K}_{|C_i}$ consist of $3$ clusters $S_1,S_2,S_3$ colored respectively by blue, orange and red. $T_4$ is a tree rooted in $x_4$ colored in bold green, that consist of 3 shortest paths. Each of $S_1,S_2,S_3$ has a vertex of $T_4$ as a neighbor.
		$R_4$ is chosen according to $\exp(1)$. The new cluster $S_4$, colored in green, consist of all vertices in $C_4$ at distance at most $R_4\Delta$ from $T_4$ w.r.t. $d_{G_4}$.
	}
\end{figure}

Provided that the graph $G$ excludes $K_r$ as a minor, for every $C_i$ it holds that $\left|{\cal K}_{|C_i}\right|\le r-2$. Indeed, by induction for every $S_{j},S_{j'}\in {\cal K}_{|C_i}$, there is an edge between $S_{j}$ to $S_{j'}$.\footnote{To see this note that there is a path between $u_j$ to $u_{j'}$ in $C_i$. Therefore, when creating $S_{j'}$ (assuming $j<j'$), it was the case that $S_j\in{\cal K}_{|C_{j'}}$. In particular, $T_{j'}$ contains a vertex with neighbor in $S_j$.}
Assume for contradiction that $\left|{\cal K}_{|C_i}\right|\ge r-1$. By contracting all the internal edges in $C_i$ and in the clusters in ${\cal K}_{|C_i}$ we will obtain $K_r$ as a minor, a contradiction.
It follows that for every $i$, $T_i$ is an $(r-2)$-core of $S_i$, with radius distributed according to exponential distribution.

Our final algorithm will be as follows: we execute \Cref{alg:coreClustering} to obtain a partition $\mathcal{S}$. 
Set $\boldsymbol{\alpha}=7r$.
For every cluster $S_i\in\mathcal{S}$ such that $R_i\ge 2\alpha$, we will (recursively) execute \Cref{alg:coreClustering} on the induced graph $G[S_i]$ to obtain an $r$-core partition with radius $2\alpha$. The final clustering $\tilde{\mathcal{S}}$ will consist of the union of all the clusters with radii at most $2\alpha$ in $\mathcal{S}$, and the clusters returned by the all recursive calls.
Clearly, $\tilde{\mathcal{S}}$ is an $r$-core partition with radius $2\alpha=O(r)\cdot\Delta$.
% In particular, \Cref{alg:coreClustering} indeed produces an $r$-core partitions with radius $\Delta$.

%
%Abraham \etal \cite{AGGNT19} called the core $T_i$ of each cluster a \emph{skeleton}. Their algorithm induce an iterative process that creates ``skeletons'' and removes their $R_i$ neighborhoods  (a
%buffer) from the graph. $R_i$ was sampled according to truncated exponential distribution. They called such an algorithm a \emph{threatening} skeleton-process.
%In general, they consider such a process where each $R_i$ is drawn according to $\Texp_{[l,u]}(\frac{b}{u-l})$, for $0=l<u\le 1$. 

For the sake of analysis, instead of discrete graphs and functions, we shall work with their continuous counterparts (see a similar approach in \cite{RR98}).
Formally, we will associate our graph with a one-dimensional simplicial complex
endowed with a metric. 
Each edge $e=(v,u)$ of weight $w(e)$ will be an interval of length $w(e)$ equipped with the standard line metric, where one endpoint identified with the vertex $v$ and the other with $u$. Naturally, the edge metrics  induce a global metric on the
entire structure. Note that the original distance between every pair of vertices is preserved (in particular this holds in every induced subgraph).
During the execution of the algorithm we can assume that in all the arbitrary choices of centers $x_i$ in \lineref{alg:CenterChoose}, the algorithm chooses a real vertex from $V$ (and once no such vertices remain, the algorithm halts). Furthermore, when choosing $u_j\in C_i$ in \lineref{alg:CenterChoose}, we will simply choose a real vertex in $V$ that has a geodesic path towards a vertex in $S_j$.
All in all, the continuous interpretation makes no difference on the execution of the algorithm, or the resulting partition, but will make the next definitions simpler and more natural.

The heart of our analysis will be to show that an arbitrary vertex $z$ belong to a cluster with large radius with probability at most $\frac13$.
Let $\mathcal{J}_{z}=\left\{T_{i}\mid d_{G_{i}}(z,S_{i})\le\alpha\Delta\right\}$ be the cores of the clusters at distance at most $\alpha\Delta$ from $z$. Alternatively, $\mathcal{J}_{z}$ consist of all the cores $T_i$ such that $d_{G_i}(T_i,z)\le R_i+\alpha\Delta$.
Note that $T_{i}$ can join $\mathcal{J}_{z}$ only if $T_i$ is in the connected component of $z$, 
%iff $S_{i}\in\mathcal{K}_{\mid C_{i}}$
and $R_{i}$ is large enough. In particular $z$ will join a cluster from $\mathcal{J}_{z}$.
The following lemma is the technical heart of our proof.

%\subsubsection{Bounding $\left|\mathcal{J}_{z}\right|$}
\begin{lemma}\label{lem:ExpectedThretaners}
	$\mathbb{E}\left[\left|\mathcal{J}_{z}\right|\right]\le \frac{2r}{\alpha+r}\cdot e^{\alpha}$. 
\end{lemma}
\begin{proof}	
We will think of growing the clusters using the exponential distribution as an iterative continues process. We define a potential function to measure the ``progress'' made twoards clustering $z$.
Consider step $i$ where $z$ is yet unclustered, and let $\widetilde{\mathcal{K}}_{C_{i}}=\left\{ S_{j}\in\mathcal{K}_{C_{i}}\mid d_{G_{i}\cup S_{j}}(z,S_{j})\le\alpha\Delta\right\}$
 be the subset of $\mathcal{K}_{C_{i}}$ clusters at distance at
most $\alpha\Delta$ from $z$.\footnote{Note that $\mathcal{J}_z$ contains (the core of) clusters that been at distance at most $\alpha\Delta$ at the time of their creation, while $\widetilde{\mathcal{K}}_{C_{i}}$ contains clusters that are at distance at most $\alpha\Delta$ in the current connected component $C_i$.}
Suppose that $\widetilde{\mathcal{K}}_{\mid C_{i}}=\left\{ S_{i_{1}},\dots,S_{i_{l}}\right\} $,
and note that necessarily $l\le r-2$. 
%Let $\boldsymbol{\lambda}=1+\frac{1}{14r}$.
Set $\boldsymbol{x}:=(x_{1},\dots,x_{l})$
where $x_{j}=\frac{d_{G_{i}\cup S_{i_{j}}}(z,S_{i_{j}})}{\Delta}$,
and 
\[
\Phi(\boldsymbol{x})=\sum_{j=1}^{l}e^{- x_{j}}
\]
For every vector \textbf{$\boldsymbol{x}$ }containing $0$ or a negative value  set $\Phi(\boldsymbol{x})=2r$ ($\Phi$ is defined on general vectors).
%This will imply also that the process halts (and that $z$ is clustered).
While $z$ is unclustered, the value $\Phi(x)$ is upper bounded by $r$.
For simplicity of notation, we will assume that $z\in C_{i}$. That is that until $z$ is clustered, \Cref{alg:coreClustering} in \lineref{line:PickCluster} always picks the cluster containing $z$. We can assume this as  otherwise the cluster $S_i$ can be ignored.
%), and also that $S_{i}\in\mathcal{J}_{z}$ (as again,
%otherwise the potential is unchanged, and $\mathcal{J}_z$ did not grow). 
Let $\tilde{R}_{i}$ be the amount $R_i$ needs to grow so that $T_i$ will join $\mathcal{J}_z$, that is $d_{G_i}(z,S_i)\le\alpha\cdot\Delta$. Formally 
\[
\tilde{R}_{i}=\begin{cases}
	\frac{d_{G_{i}}(z,T_{i})}{\Delta}-\alpha & \text{if }d_{G_{i}}(z,T_{i})>\alpha\cdot\Delta\\
	0 & \text{if }d_{G_{i}}(z,T_{i})\le\alpha\cdot\Delta
\end{cases}~.
\]
Let $h=\frac{1}{\Delta}\cdot d_{G_{i}}(z,T_{i})-\tilde{R}_{i}$,
be the amount that $R_{i}$ need to additionally grow so that $z$
will join $S_{i}$. 
Note that $h$ either equals $\alpha$ (if $d_{G_{i}}(z,T_{i})>\alpha\cdot\Delta$) or to $\frac{d_{G_{i}}(z,T_{i})}{\Delta}$. In any case, $h\le\alpha$.
See \Cref{fig:RtildeDef} (left) for an illustration.

%if $d_{G_{i}}(z,T_{i})\ge\alpha\Delta$
%then $\tilde{R}_{i}=\frac{d_{G_i}(z,T_{i})}{\Delta}-\alpha$
%and $h=\alpha$, while if $d_{G_i}(z,T_{i})<\alpha\Delta$ then $\tilde{R}_{i}<0$
%and $h=\frac{1}{\Delta}\cdot d_{G_{i}}(z,T_{i})$. In any case, $h\le\alpha$.
%See \Cref{fig:RtildeDef} (left) for an illustration.
%
%
%value s.t. $d_{G_i}(z,T_{i})=(\alpha+\tilde{R}_i)\Delta$ (formally $\tilde{R}_i=\frac{d_{G_i}(z,T_i)}{\Delta}-\alpha$, note that $\tilde{R}_i$
%might be negative), and $\tilde{A}_{i}=B_{G_{i}}(T_{i},\max\{\tilde{R}_i,0\})$
%be the cluster $S_{i}$ at the time where $T_{i}$ joins $\mathcal{J}_{z}$.

\begin{figure}[]
	\centering{\includegraphics[scale=0.65]{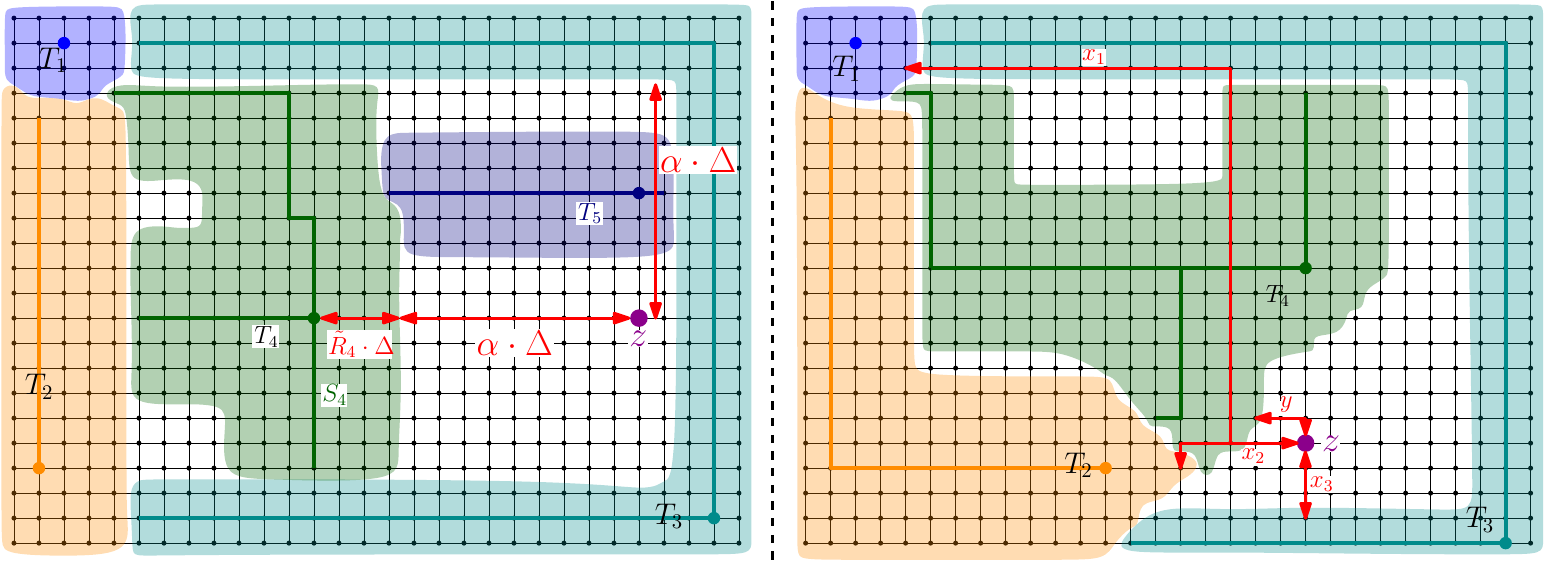}} 
	\caption{\label{fig:RtildeDef}\small \it 
		In both figures illustrated an unweighted graph, the clusters $S_1,S_2,S_3$ are colored in blue, orange and cyan respectively.
		At the forth step a new cluster is created with core $T_4$.\\
		On the left, $z$ is a vertex at distance greater than $\alpha\cdot\Delta$ from $T_4$. Hence $\tilde{R}_4$ is defined to be $\tilde{R}_i=\frac{d_{G_4}(z,T_4)}{\Delta}-\alpha$. 
%		$\tilde{A}_i$ is set to be $B_{G_{i}}(T_{i},\tilde{R}_i)$. 
		If $R_i$ is greater than $\tilde{R}_i$, then $T_i$ will join $\mathcal{J}_z$, and we will set $h=\alpha$. In the figure $R_4$ equals $\tilde{R}_4$.
		Next, the algorithm creates a cluster $S_5$ with core $T_5$, where $d_{G_5}(z,T_5)\le \alpha\cdot \Delta$.
		Hence $\tilde{R}_5$ is set to be $0$, $T_5$ joins $\mathcal{J}_z$, and $h$ set to be $\frac{1}{\Delta}\cdot d_{G_{5}}(z,T_{5})$.\\
		On the right, just before the creation of $S_4$ the vector $\boldsymbol{x}=(x_1,x_2,x_3)$ where $x_1=34$, $x_2=6$, and $x_3=3$ consist of the distances from $z$ to the clusters $S_1,S_2,S_3$ respectively. 
		The potential at this point is $\Phi(\boldsymbol{x})=e^{-3}+e^{-6}+e^{-34}$.
		After the creation of $S_4$ with $y=3$, there is no longer a path from $S_1$ to $z$, the distance from $S_2$ to $z$ increased from $6$ to $8$, and the distance from $S_3$ to $z$ remained unchanged. The new vector is  $\boldsymbol{x}'=(8,3,3)$, and the potential is $\Phi(\boldsymbol{x}')=e^{-3}+e^{-3}+e^{-8}$. 
	}
\end{figure}

If $R_{i}<\tilde{R}_{i}$, then no new cluster joins $\mathcal{J}_z$, and the potential remains unchanged. Thus nothing happened from our perspective. We will thus assume that  $R_{i}\ge\tilde{R}_{i}$. Let $\widehat{R}_{i}=R_{i}-\tilde{R}_{i}$. By the memoryless property, $\widehat{R}_{i}$ is distributed according to exponential distribution. Let $y=h-\widehat{R}_{i}$.
Then if $R_i<d_{G_i}(T_i,z)$ then $y=d_{G_i}(S_i,z)$, and otherwise $y$ has a negative value (and $z$ is clustered).
We next analyze how the vector $\boldsymbol{x}$ can change as a result of creating $S_i$.
For every $j$, if $x_{j}\le y$, then the shortest path from $z$ to $S_{i_{j}}$
is completely disjoint from $S_{i}$. In particular $x_{j}$ will
remain unaffected. Otherwise, if $x_{j}> y$ then some vertices in the shortest path from $S_{i_j}$ to $z$ might join the new cluster $S_i$. It follows that $x_j$ can 
% Let $j^{*}$ be the maximal index s.t. $x_{j}\le y$.
%Then the first $j^{*}$ indices in $\boldsymbol{x}$ will remain unchanged,
%the $j^{*}+1$ index becomes $y$, and all other can either
increase, disappear or remain unchanged.
See \Cref{fig:RtildeDef} (right) for an illustration.

For a sequence of numbers $\boldsymbol{x}$ denote by $\boldsymbol{x}\downarrow y$
the sequence where we delete all the numbers larger than $y$, and
add $y$. Let $\boldsymbol{x}(i)$ be the sequence at time $i$.
We analyze the change in the potential function. The expectations
are bounded from below by $\mathbb{E}\left[\Phi(\boldsymbol{x}(i+1))-\Phi(\boldsymbol{x}(i))\right]\ge\mathbb{E}\left[\Phi(\boldsymbol{x}\downarrow (h-\widehat{R}_{i}))-\Phi(\boldsymbol{x})\right]$. 
The inequality holds as in $\boldsymbol{x}\downarrow (h-\widehat{R}_{i})$ we nullify all the coordinates smaller than $h-\widehat{R}_{i}$, while in $\boldsymbol{x}(i+1)$ some of them might remain unchanged, or only increased.
\begin{claim}\label{clm:expectedGain}
	For every vector $\boldsymbol{x}\in(0,\alpha]^{l}$ for $l\in[0,r]$,
	for every $h\in[0,\alpha]$ and $\rho$ distributed according to $\exp(1)$, it holds that that $\mathbb{E}\left[\Phi(\boldsymbol{x}\downarrow h-\rho)-\Phi(\boldsymbol{x})\right]\ge e^{-\alpha}\cdot(\alpha+r)$.
\end{claim}

\begin{proof}
	If $h=0$, then $\Phi(\boldsymbol{x}\downarrow h-\rho)-\Phi(\boldsymbol{x})\ge2r-r\ge e^{-\alpha}\cdot(\alpha+r)$. Thus we can assume that $h>0$. The gain to
	$\Phi$ is a single new addend of value $e^{-(h-\rho)}$ if $\rho<h$, or $2r$ in case $\rho\ge h$. Specifically:
	\begin{align*}
		\mathbb{E}\left[\mbox{gain}\right] & =\int_{0}^{h}e^{-\rho}\cdot e^{-(h-\rho)}d\rho+\int_{h}^{\infty}e^{-\rho}\cdot2rd\rho\\
		& =e^{-h}\int_{0}^{h}1d\rho+2r\cdot e^{-h}=e^{-h}\left(h+2r\right)~.
	\end{align*}	
%	The integral is only until $h$, as afterwards the potential becomes $\boldsymbol{\beta}\cdot s$. 
	From the other hand, for the addend $x_{j}$, there is a loss of $e^{-x_{j}}$ if $h-\rho<x_{j}$, and no change otherwise. 
	If $h<x_{j}$, this happens with probability $1$ and the loss is $e^{-x_{j}}<e^{-h}$.
	Else, the expected loss is
	\[
	\Pr\left[h-\rho<x_{j}\right]\cdot e^{-x_{j}}=e^{-(h-x_{j})}\cdot e^{-x_{j}}=e^{-h}~.
	\]
Thus the total expected loss is bounded by $r\cdot e^{-h}$. In total
\[
\mathbb{E}\left[\Phi(\boldsymbol{x}\downarrow h-\rho)-\Phi(\boldsymbol{x})\right]\ge e^{-h}\left(h+2r-r\right)=e^{-h}\cdot\left(h+r\right)~.
\]
The function $e^{-h}\cdot(h+r)$ monotonically decreasing, and has minimum value at $h=\alpha$. The claim now follows.
\end{proof}

Set $\zeta=e^{-\alpha}\cdot(\alpha+r)$. Then it follows from the claim above,
that for every $i$, and $\boldsymbol{x}(i)$, it holds that $\mathbb{E}\left[\Phi(\boldsymbol{x}(i+1))-\Phi(\boldsymbol{x}(i))\mid\Phi(\boldsymbol{x}(i))\right]\ge\zeta$.
Let $X_{t}=\Phi(\boldsymbol{x}(t))-t\zeta$. It holds that
\begin{align*}
	\mathbb{E}\left[X_{t+1}\mid X_{1}\dots X_{t}\right] & =\mathbb{E}\left[\Phi(\boldsymbol{x}(t+1))-(t+1)\zeta\mid\boldsymbol{x}(t)\right]\\
	& \ge\Phi(\boldsymbol{x}(t))+\zeta-(t+1)\zeta=\Phi(\boldsymbol{x}(t))-t\zeta=\mathbb{E}\left[X_{t}\right]~,
\end{align*}
thus $X_{1},X_{2},\dots,$ is a sub-martingale. Recall that $\Phi(\boldsymbol{x})$
is always bounded by $2r$. Denote by $\tau=|\mathcal{J}_{z}|$
the time when the process halts. Using Doob's optional stopping time
theorem \cite{BW07book} it holds that 
\[
\mathbb{E}\left[\Phi(\boldsymbol{x}(\tau))\right]-\mathbb{E}[\tau]\cdot\zeta=\mathbb{E}\left[X_{\tau}\right]\ge\mathbb{E}\left[X_{0}\right]=0
\]
Implying $\mathbb{E}[\tau]\le\frac{1}{\zeta}\cdot\mathbb{E}\left[\Phi(\boldsymbol{x}(\tau))\right]=\frac{2r}{\zeta}=\frac{2r}{\alpha+r}\cdot e^{\alpha}$.
\end{proof}

Finally we are ready to bound the probability that $z$ joins to a cluster of too large 	radius.
Let $\Psi_{z}$ be the event that $z$ joins a cluster $S_{i}$,
such that at the time that $z$ joins, $R_{i}>\alpha$ (formally, $z\in S_i$ and $d_{G_i}(z,T_{i})>\alpha\Delta$). 
%Recall that $\Psi_{z}$ is the event that $z$ joins a cluster $S_{i}$, such that at the time that $z$ joins, $R_{i}>\alpha$.
Set $\mathcal{T}_{z}=\left\{ T_{i}\mid d_{G_{i}}(z,T_{i})>\alpha\Delta\text{ and }d_{G_{i}}(z,S_{i})\le\alpha\Delta\right\} $.
Denote by $\mathcal{C}_{i}$ the event that $z\in S_{i}$ and $T_{i}\in\mathcal{T}_{z}.$
Recall that $\mathcal{J}_{z}=\left\{ T_{i}\mid d_{G_{i}}(z,S_{i})\le\alpha\Delta\right\} $,
thus $\mathcal{T}_{z}\subseteq\mathcal{J}_{z}$. Hence by \Cref{lem:ExpectedThretaners},
$\mathbb{E}\left[\left|\mathcal{T}_{z}\right|\right]\le\mathbb{E}\left[\left|\mathcal{J}_{z}\right|\right]\le\frac{2r}{\alpha+r}\cdot e^{\alpha}$. We conclude,
\begin{align*}
	\Pr\left[\Psi_{z}\right] & =\Pr\left[\cup_{i}\mathcal{C}_{i}\right]=\sum_{i}\Pr\left[\mathcal{C}_{i}\right]=\sum_{i}\Pr\left[z\in S_{i}\wedge T_{i}\in\mathcal{T}_{z}\right]\\
	& =\sum_{i}\Pr\left[z\in S_{i}\mid T_{i}\in\mathcal{T}_{z}\right]\cdot\Pr\left[T_{i}\in\mathcal{T}_{z}\right]\\
	& \overset{(*)}{=}\sum_{i}e^{-\alpha}\cdot\Pr\left[T_{i}\in\mathcal{T}_{z}\right]\\
	& =e^{-\alpha}\cdot\mathbb{E}\left[\left|\mathcal{T}_{z}\right|\right]\le e^{-\alpha}\cdot\frac{2r}{\alpha+r}\cdot e^{\alpha}=\frac{2r}{\alpha+r}=\frac{1}{4}~,
\end{align*}
where in the equality $(*)$ we used the memoryless property of exponential
distribution.

Let $\varUpsilon_{z}$ be the event that $z$ joined a cluster $S_{i}$,
such that after the time that $z$ joins, $R_{i}$ increases by additional $\alpha$ factor  (formally, $z\in S_i$ and $R_i\cdot\Delta-d_{G_i}(z,T_{i})>\alpha\Delta$). Then by the memoryless property of exponential distribution,
$\Pr[\varUpsilon_{z}]=e^{-\alpha}<\frac{1}{4}$.
Denote by $\Phi_{z}$ the event that $z$ belong to a cluster with
radius greater than $2\alpha\Delta$. 
Clearly $\Phi_{z}\subseteq \Psi_z\cup\varUpsilon_{z}$, as if both $\Psi_z,\varUpsilon_{z}$ did not occur then $R_i<2\alpha$. By union bound it follows that $\Pr\left[\Phi_{z}\right]\le\Pr\left[\Psi_{z}\right]+\Pr\left[\varUpsilon_{z}\right]\le\frac{1}{4}+\frac{1}{4}=\frac{1}{2}$.
%
%\[
%\Pr\left[\Phi_{z}\right]\le\Pr\left[\Psi_{z}\right]+e^{-\alpha}\le\frac{1}{3}~.
%\]

Next we prove the padding property. Let $\gamma\in(0,\frac18)$ be a padding parameter, and set $B=B_G(z,\gamma\Delta)$. 
First we argue that in a single partition executed using \Cref{alg:coreClustering}, the ball $B$ is fully contained in a single cluster with probability at least $e^{-2\gamma}$.
Indeed consider the first index $i$ such that $B\cap S_i\ne\emptyset$, and let $u$ be the closest vertex to $T_i$ (that is $u={\rm arg}\min_{v\in V}d_{G_i}(v,T_i)$).
As $B\subseteq G_i$, for every $v\in B$ it holds that $d_{G_i}(u,v)\le2\gamma\cdot\Delta$. Hence by the memoryless property of exponential distribution 
\begin{align*}
\Pr\left[B\subseteq P(z)\right] & =\Pr\left[B\subseteq S_{i}\mid u\in S_{i}\right]\\
& \ge\Pr\left[R_{i}\ge d_{G_{i}}(T_{i},u)+2\gamma\mid R_{i}\ge d_{G_{i}}(T_{i},u)\right]=\Pr\left[R_{i}\ge2\gamma\right]=e^{-2\gamma}~.
\end{align*}

Recall that our algorithm returns a partition $\tilde{\mathcal{S}}$, which constructed by first executing \Cref{alg:coreClustering} to obtain a partition $\mathcal{S}$, and then recursively partitioning each cluster $S_i\in \mathcal{S}$ with radius larger than $2\alpha\cdot\Delta$. Let $\tilde{P}(z)$ be the cluster containing $z$ in the final partition $\tilde{\mathcal{S}}$.
Denote by $\varGamma_{i}$ the event that $z$ participated in at least $i$ recursive calls, where in the first $i-1$ recursive calls the ball $B$ was contain in a single cluster, and in the $i$'th recursive call, the ball $B$ was cut. Then $\Pr[\varGamma_{i}]\le\frac{1}{2^{i-1}}\cdot(1-e^{-2\gamma})$, as for $\varGamma_{i}$ to occur, $z$ must be clustered $i-1$ times in a cluster with radius greater than $2\alpha\cdot\Delta$ (each occurring with probability at most $\frac12$), and in the $i$'th iteration the ball $B$ must be cut (which happens with probability at most $1-e^{-2\gamma}$). We conclude

%
%Denote by $P_{1}(z)$ the cluster containing $z$ after the first
%iteration. Similarly $P_{i}(z)$. Note that if $P_{i}(z)$ has diameter
%smaller than $\alpha\Delta$, then $P_{i}(z)=P_{i+1}(z)=\cdots$.
%Denote by $P^{\infty}(z)$ the final cluster containing $z$. Denote
%by $\varGamma_{i}$ the event that $z$ was active in the $i$ iteration,
%the ball $B_{G}(z,\gamma\Delta)\subseteq P^{i-1}(z)$, and $B_{G}(z,\gamma\Delta)\nsubseteq P^{i}(z)$.
%Clearly, $\Pr\left[B_{G}(z,\gamma\Delta)\nsubseteq P^{\infty}(z)\right]=\sum_{i\ge1}\Pr\left[\Gamma_{i}\right]$.
%For $\Gamma_{i}$ to occur, $\Phi_{z}$ must occur $i-1$ consecutive
%times, and then the ball $B_{G}(z,\gamma\Delta)$ should not belong
%to a single cluster. Using memoryless, the probability that in a single
%clustering step the ball $B_{G}(z,\gamma\Delta)$ do not belong to
%a single cluster is bounded by $\int_{0}^{2\gamma}e^{x}dx=1-e^{-2\gamma}$
%(because once some vertex from $B_{G}(z,\gamma\Delta)$ joins some
%cluster, if the radius will increase by more than $2\gamma$, then
%all the vertices in $B_{G}(z,\gamma\Delta)$ will join that cluster.
%It follows that $\Pr\left[\Gamma_{i}\right]\le\frac{1}{3^{i-1}}\cdot\left(1-e^{-2\gamma}\right)$.
%As the events $\varGamma_{i}$ are disjoint, 
%We have
\[
\Pr\left[B_{G}(z,\gamma\Delta)\nsubseteq\tilde{P}(z)\right]\le\sum_{i\ge1}\Pr\left[\varGamma_{i}\right]\le\sum_{i\ge1}\frac{1}{2^{i-1}}\cdot\left(1-e^{-2\gamma}\right)\le2\cdot2\gamma\le1-e^{-8\gamma}~.
\]

To conclude we constructed $r$-core partitions with radius $2\alpha\cdot\Delta$ such that for every vertex $z\in V$ and $\gamma\in(0,\frac18)$ it holds that $\Pr
\left[B_{G}(z,\gamma\Delta)\subseteq \tilde{P}(z)\right]\ge e^{-8\gamma}$. \Cref{lem:coreClustering} now follows by scaling.

\section{General Graphs}\label{sec:generalGraphs}
Bartal \cite{Bar96} showed that every $n$-point metric space is $O(\log n)$-decomposable. In particular, for large enough constant $c$, and $\gamma=\frac{1}{c\cdot k}$, a pair of vertices at distance $\gamma\Delta$ will be clustered together with probability $n^{-\frac1k}$. 
As the padding parameter governs the exponent in the success probability, it is important to optimize the constant $c$. 
Awerbuch and Peleg \cite{AP90} showed that for $k\in\N$, general $n$-vertex graphs admit a strong $(4k-2,O(k\cdot n^{\frac1k}))$ sparse cover scheme. Specifically, \cite{AP90} gave a deterministic construction of $O(k\cdot n^{\frac1k})$ partitions, all strongly $\Delta$-bounded and such that every vertex is $4k-2$-padded in one of these partitions. It follows that if one samples a single partition from \cite{AP90} uniformly at random, then every vertex is $4k-2$-padded with probability at least  $\Omega(\frac1k\cdot n^{-\frac1k})$.

In this section we attempt to optimize the ratio trade-off of ThProbabilistic partitions.
We will prove that every $n$-vertex weighted graph admits a strong $(2k,\frac18\cdot n^{-\frac{1}{k-1}})$-ThProbabilistic decomposition scheme (\Cref{thm:PartitionStrongMetrics}), and weak $(2k,n^{-\frac{1}{k}})$-ThProbabilistic decomposition scheme (\Cref{thm:PartitionWeakMetrics}).
Finally, we show that assuming Erd\H{o}s girth conjecture \cite{Erdos64}, for every (integer) $k\ge 1$, and (real) $t<2k+1$, if every $n$ point metric space admits a weak $(t,p)$-ThProbabilistic decomposition scheme, then $p=\tilde{O}(n^{-\frac{1}{k}})$.
It follows that the success probability in \Cref{thm:PartitionWeakMetrics} cannot be (substantially) improved.
However, note that it might be possible to improve the stretch parameter to $2k-1$ (instead of $2k$) while having the same success probability.
Closing this gap between $2k-1$ to $2k$, and constructing strong ThProbabilistic decomposition that will match the performance of the weak ThProbabilistic decompositions, are two intriguing open problems.

%
%From the other hand, assuming Erd\H{o}s girth conjecture \cite{Erdos64}, one can construct an unweighted graph such that in every $\Delta$-bounded partition for $\Delta<2k+1$, for some edge $(u,v)$ the probability that $u,v$ are clustered together is at most $\tilde{O}(n^{-\frac1k})$. \todo{ref to proof?}
%\footnote{The girth of a graph is the length of its shortest cycle. Erd\H{o}s girth conjecture \cite{Erdos64} states that there is a graph $G$ with girth $2k+1$ and $\Omega(n^{1+\frac1k})$ edges. Note that for every edge $e=\{u,v\}$ in $G$, the distance between $u,v$ in $G\setminus\{e\}$ is at least $2k$. Consider a subgraph $G'$ of $G$, and let $\mathbb{D}'$ be a distance oracle for $G'$ distinguishing between the cases $d_{G'}(u,v)\le 1$ and $d_{G'}(u,v)>2k-1$. Thus every two different subgraphs $G_1,G_2$ should have different distance oracles. As there are $2^{\Omega(n^{1+\frac1k})}$ subgraphs, there are $2^{\Omega(n^{1+\frac1k})}$ different distance oracles, and in particular some subgraph $G'$, the distance oracle of which must have size $\Omega(n^{1+\frac1k})$.
%}

\subsection{Strong Diameter for General Graphs}

\begin{theorem}\label{thm:PartitionStrongMetrics}
	For every (real) $k\ge 1$, every $n$ point weighted graph $G=(V,E,w)$ admits a strong $(2k,\frac18\cdot n^{-\frac{1}{k-1}})$-ThProbabilistic decomposition scheme.
\end{theorem}
\begin{proof}
	Consider an $n$ point weighted graph $G=(V,E,w)$. We will show that $G$ admits a strong $(2k,2^{-\frac{k}{k-1}}\cdot n^{-\frac{1}{k-1}})$-ThProbabilistic decomposition scheme. For $k\ge\frac32$, $2^{-\frac{k}{k-1}}\ge\frac18$, so the theorem follows. For $k<\frac32$, $n^{-\frac{1}{k-1}}<n^{-2}$. Note that every graph admits a strong $(1,n^{-2})$-ThProbabilistic decomposition scheme.%
	\footnote{Consider the following random partition: pick u.a.r. a pair of vertices $u,v$ at distance at most $\Delta$. $\mathcal{P}$ will consist of a set containing the shortest path from $u$ to $v$, and all the remaining vertices will be in singleton clusters. Clearly this is a strong $(1,n^{-2},\Delta)$-ThProbabilistic decomposition.}
%	\footnote{Here is a way to get weak $(1,n^{-1})$-ThProbabilistic decomposition scheme. It is  known that the edges of the $n$-clique $K_n$ can be partitioned into 
%	$n$ disjoint matchings ($n-1$ for even $n$). This is follows for example from \cite{Ber78}.
%	Create a decomposition by sampling at random one of these matchings of $K_n$. For every pair $\{u,v\}$ in the matching we create a cluster containing both $u,v$ if $d_G(u,v)\le \Delta$, or simply add to singletons $\{u\},\{v\}$.}
	Hence there is nothing to prove.
	
	By scaling, it is enough to construct a ThProbabilistic decomposition for $\Delta=2$.
	We will use a classic ball carving (ala \cite{Bar96}).
	Fix $\lambda=\frac{k}{k-1}\cdot\ln(2n)$, and let $\mathcal{D}$ be exponential distribution with parameter $\lambda$. The density function is $\lambda\cdot e^{-\lambda x}$ for $x\ge 0$. Note that by our choice of $\lambda$, it holds that 
	\begin{align}
		e^{-\frac{\lambda}{k}}=2n\cdot e^{-\lambda}~.\label{eq:LambdaChoise}
	\end{align}
%	Indeed, the equation holds iff $2\delta\cdot \lambda=\lambda-\ln(2n)$
%	which in turn holds for $\lambda=\frac{\ln(2n)}{1-2\delta}$.
	%	
	We create a partition $\mathcal{P}$ as follows, initially $Y_1=V$ is all the unclustered vertices. At step $i$, after we created the clusters $C_1,C_2,\dots,C_{i-1}$, the unclustered vertices are $Y_i=V\setminus(\cup_{j<i}C_j)$. Pick an arbitrary vertex $x_i\in Y_i$, and a radius $r_i\sim\mathcal{D}$.  Set the new cluster to be $C_{i}=B_{G[Y_i]}(x_i,r_i)$ the ball of radius $r_i$ around $x_i$ in the graph induced by the unclustered vertices. The process halts once $Y_i=\emptyset$.
	
	%	let $C_1,\dots,C_m$ be the clusters created so far. While $\cup_{i\le m}C_i\ne X$, pick an arbitrary point $v\notin\cup_{i\le m}C_i$, and a radius $r_v\sim\mathcal{D}$.  Set $C_{m+1}=B_X(v,r_v)\setminus\cup_{i\le m}C_i$. The process halts when all the vertices belong to some cluster.
	
	Consider a pair $u,v$ such that $d_G(u,v)=\frac{1}{2k}\cdot\Delta=\frac1k$. 
	We first argue that $\Pr[P(x)=P(y)] \ge e^{-\frac\lambda k}$.
	Let $Q$ be some shortest path in $G$ from $x$ to $y$.
	Let $i$ be the first index such that $C_i\cap Q\ne\emptyset$.
	In particular, $Q\subseteq Y_i$ and hence $d_{G[Y_i]}(u,v)=d_{G}(u,v)$.
	Let $z\in Q$ be the vertex closest to $x_i$ (that is $z=\mbox{arg}\min_{s\in Q}d_{G[Y_i]}(x_i,s)$).
	Hence $d_{G[Y_i]}(x_i,z)\le r_i$. By triangle inequality, for every $s\in Q$, $d_{G[Y_i]}(x_i,s)-d_{G[Y_i]}(x_i,z)\le d_{G[Y_i]}(z,s)\le\frac1k$. Using the memoryless property of exponential distribution we conclude
	\begin{align*}
		& \Pr\left[P(u)=P(v)\mid Q\subseteq Y_{i}\text{ and }Q\cap C_{i}\ne\emptyset\right]\\
		& \qquad\ge\Pr_{r_{i}\sim\mathcal{D}}\left[r_{i}\ge\max\left\{ d_{G[Y_{i}]}(x_{i},v),d{}_{G[Y_{i}]}(x_{i},u)\right\} \mid r_{i}\ge d{}_{G[Y_{i}]}(x_{i},z)\right]\\
		& \qquad=\Pr_{r_{i}\sim\mathcal{D}}\left[r_{i}\ge\max\left\{ d_{G[Y_{i}]}(x_{i},v)-d_{G[Y_{i}]}(x_{i},z),d{}_{G[Y_{i}]}(x_{i},u)-d_{G[Y_{i}]}(x_{i},z)\right\} \right]\\
		& \qquad\ge\Pr_{r_{i}\sim\mathcal{D}}\left[r_{i}\ge\frac{1}{k}\right]=e^{-\frac{\lambda}{k}}~.
	\end{align*}
	
	If all the radii sampled in the process will be at most $1$, our partition will be strongly $2$-bounded. The probability of a single radius of being larger then $1$ is $e^{-\lambda}$. By union bound the probability that some radius is larger then $1$ is at most $n\cdot e^{-\lambda}$. Hence the probability that the partition is $2$ bounded is at least $1-n\cdot e^{-\lambda}$. We will condition the resulting partition to be $2$ bounded.
	Using the law of total probability, the probability that $P(u)=P(v)$ now is:
	\begin{align*}
		& \Pr[P(u)=P(v)\mid\mathcal{P}\text{ is }2\text{-bounded}]\\
		& \qquad=\frac{\Pr[P(u)=P(v)]}{\Pr[\mathcal{P}\text{ is }2\text{-bounded}]}-\frac{\Pr[\mathcal{P}\text{ is not }2\text{-bounded}]}{\Pr[\mathcal{P}\text{ is }2\text{-bounded}]}\cdot\Pr[P(u)=P(v)\mid\mathcal{P}\text{ is not }2\text{-bounded}]\\
		& \qquad\stackrel{(*)}{>}\Pr[P(u)=P(v)]-\Pr[\mathcal{P}\text{ is not }2\text{-bounded}]\\
		& \qquad\ge e^{-\frac{\lambda}{k}}-n\cdot e^{-\lambda}=\frac{1}{2}\cdot e^{-\frac{\lambda}{k}}=\frac{1}{2}\cdot(2n)^{-\frac{1}{k-1}}=2^{-\frac{k}{k-1}}\cdot n^{-\frac{1}{k-1}}~.
	\end{align*}	
	where inequality $^{(*)}$ follows as $\Pr[\mathcal{P}\text{ is }2\text{-bounded}]<1$ and ${\Pr[P(u)=P(v)\mid\mathcal{P}\text{ is not }2\text{-bounded}]\le1}$.
\end{proof}

\subsection{Weak diameter for General Metrics}
Here we present a different clustering scheme with improved parameters, however the diameter guarantee will be only weak. 
The algorithm here is ball carving as well, however instead of sampling i.i.d. random radii for arbitrarily chosen centers, we pick a single random radius, and the order of centers is chosen randomly. This approach was introduced by C{\u{a}}linescu, Karloff and Rabani \cite{CKR04} in the context of the $0$-extension problem, then Fakcharoenphol, Rao, and Talwar \cite{FRT04} used it to construct stochastic embedding into ultrametrics (alternatively: creating random hierarchical partitions). In particular, one can deduce from \cite{FRT04} analysis an $(O(\log n),O(1)))$-ThProbabilistic decomposition scheme.
This approach was also the first to construct weak padded decompositions for doubling metrics \cite{GKL03}.
Finally,  Blelloch, Gu, and Sun \cite{BGS17} used such clustering to construct Ramsey trees. Furthermore, one can deduce from \cite{BGS17} analysis an $(O(k),O(n^{-\frac1k})))$-ThProbabilistic decomposition scheme.
Our contribution here is tighter analysis, improving the stretch constant from $O(k)$ to $2k$.

%Note that this approach inherently produces clusters with weak diameter. While the exponential balls approach could be made with strong diameter without significant changes, this approach is inherently weak diameter. Indeed, if we will take balls in the induced metric space, all our initial order could be change, and our reasoning will no longer hold.
\begin{theorem}\label{thm:PartitionWeakMetrics}
	For every (integer) $k\ge 1$, every $n$ point metric space $(X,d_X)$ admits a weak $(2k,n^{-\frac{1}{k}})$-ThProbabilistic decomposition scheme.
\end{theorem}
\begin{proof}	
	Pick u.a.r. a  radius $r\in\{\frac1k,\frac2k,\dots,\frac kk\}$, and a random permutation $\pi=\{x_1,x_2,\dots,x_n\}$ over the metric points. 
	Each point $x\in X$ joins the cluster of the first center w.r.t. $\pi$ at distance at most $r\cdot\frac\Delta2$ from $x$. Formally:
	$$C_i=B_X(x_i,r\cdot\frac{\Delta}{2})\setminus\cup_{j<i}B_X(x_j,r\cdot\frac{\Delta}{2})~.$$
	As a result we obtain a $\Delta$ bounded partition.
	%We will begin with a discrete argument, and see later if we can gain something by making it continues.
	%Fix integer $k\ge 2$, and suppose that we pick u.a.r. $r\in\{\frac1k,\frac2k,\dots,\frac kk\}$, each one with probability $\frac1k$. 
	Fix a pair of points $u,v$ such that $d_X(u,v)\le\frac{1}{2k}\cdot\Delta$. 
%	Order the vertices by increasing distance from $\{u,v\}$: $x_1,x_2,\dots,x_n$, braking ties arbitrarily. Note that $\{x_1,x_2\}=\{u,v\}$.
	Let $A_s=\{x\in X\mid d_X(x,\{u,v\})\le\frac{s}{k}\cdot\frac\Delta2\}$ be the set of points at distance at most $\frac{s}{k}\cdot\frac\Delta2$ from either $u$ or $v$. 
	Then $A_0=\{u,v\}$.
	Suppose that $r=\frac sk$, and let $x_i$ be the vertex with minimal index such that $d_X(\{u,v\},x_i)\le\frac sk\cdot\frac\Delta2$. Then $u$ and $v$ will not join the clusters $C_1,\dots,C_{i-1}$, and at least one of them will join $C_i$.
	Assume w.l.o.g. that $d_X(u,x_i)\le d_X(v,x_i)$, and suppose farther that $x_i\in A_{s-1}$. 
	By the triangle inequality it follows that $d_X(v,x_i)\le d_X(v,u)+d_X(u,x_i)\le\frac1k\cdot\frac\Delta2+\frac{s-1}{k}\cdot\frac\Delta2=\frac{s}{k}\cdot\frac\Delta2$. Hence both $u,v$ will join the cluster of $x_i$. Using the law of total probability we conclude
	\begin{align*}
		\Pr[P(u)=P(v)] & =\frac{1}{k}\cdot\sum_{s=1}^{k}\Pr[P(u)=P(v)\mid r=\frac{s}{k}]\\
		& \ge\frac{1}{k}\cdot\sum_{s=1}^{k}\frac{|A_{s-1}|}{|A_{s}|}\ge\left(\Pi_{s=1}^{k}\frac{|A_{s-1}|}{|A_{s}|}\right)^{\frac{1}{k}}=\left(\frac{|A_{0}|}{|A_{k}|}\right)^{\frac{1}{k}}\ge n^{-\frac{1}{k}}~,
	\end{align*}
	where the second inequality follows by the inequality of arithmetic and geometric means.
	
%	There must be an index $i\in[k]$ such that $|A_i|\le (\frac n2)^{\frac1k}\cdot |A_{i-1}|$, as otherwise
%	\[
%	|A_{k}|>n^{\frac{1}{k}}\cdot|A_{k-1}|>n^{\frac{2}{k}}\cdot|A_{k-2}|>\cdots>n^{\frac{k}{k}}\cdot|A_{0}|=2n~,
%	\]
%	a contradiction. 
%	Let $\Phi$ be the event $r=\frac ik$, and $\Psi$ be the event that in the permutation $\pi$ restricted to $\cup_{j\le i}A_j$ the first index is among $\cup_{j< i}A_j$. 
%	We first argue that if both $\Phi,\Psi$ occur then $u,v$ will belong to the same cluster....
%	\todo[inline]{fill}
%	
%	Next we analyze the probability that $\Phi\wedge\Psi$ occur. Clearly, $\Psi,\Phi$ are independent, and $\Pr[\Phi]=\frac1k$. Further, 
%	\[
%	\Pr[\Psi]=\frac{|\cup_{j<i}A_{i}|}{|\cup_{j\le i}A_{i}|}\ge\frac{|A_{i-1}|}{|A_{i-1}\cup A_{i}|}\ge\frac{|A_{i}|}{n^{\frac{1}{k}}\cdot|A_{i}|+|A_{i}|}\ge\frac{1}{2}\cdot n^{-\frac{1}{k}}~.
%	\]
%	We conclude $\Pr[\Phi\wedge\Psi]\ge\frac{1}{2k}\cdot n^{-\frac1k}$. 
%	\cite{CKR01}\todo{Add doi}

\end{proof}

%\atodoin{what happens with the continues choice of $r$?}
%\begin{remark}  
%	One might hope that by sampling the radius $r$ u.a.r. from the interval $[0,1]$ (instead of the discrete set  $\{\frac1k,\frac2k,\dots,\frac kk\}$), each pair $u,v$ will be clustered together with probability $n^{-\frac{2\cdot d_X(u,v)}{\Delta}}$ regardless of $d_X(u,v)$. Unfortunately, it is not the case.
%	Indeed, consider the following example: $\Delta=2$, and all the points lie on the real line where $u$ lies at $0$, $v$ at $-\frac1k$, and for every $x\ge\frac1k$, there are (about) $n^{\frac{k}{k-1}(x-\frac{1}{k})}$ points at up to $x$. 
%	For $r<\frac1k$ clearly $u,v$ will go to different clusters, and in general, for $r=x$,  $u$ and $v$ will be clustered together with probability 
%	$\frac1k$. In total we have less than $2n$ points, and the probability that $u$ and $v$ are clustered together is at most  
%\end{remark}

\subsection{Lower Bound}

\begin{theorem}\label{thm:PartitionMetricsLB}
	Assuming Erd\H{o}s girth conjecture \cite{Erdos64}, for every (integer) $k\ge 1$, and (real) $t<2k+1$, if every $n$ point metric space admits a weak $(t,p)$-ThProbabilistic decomposition scheme, then $p=\tilde{O}(n^{-\frac{1}{k}})$.
\end{theorem}
\begin{proof}
	Given a metric space $(X,d_X)$, an $(a,b)$-\emph{gap-distance-oracle}, is a data structure that given a pair of points $x,y$
	returns yes if $d_X(x,y)\le a$, no if $d_X(x,y)> b$, and can return either yes or no if $d_X(x,y)\in(a,b]$.
	Thorup and Zwick \cite{TZ05} showed that 
	assuming Erd\H{o}s girth conjecture  \cite{Erdos64}, a $t$-distance oracle (a data structure that returns $t$-approximation of the distance), must have size $\Omega(n^{1+\frac1k})$. Implicitly, their proof implies also that a gap distance oracle requires $\Omega(n^{1+\frac1k})$ space. 
	Later, we will show how using ThProbabilistic decomposition one can construct a gap distance oracle, and conclude the theorem.
	\begin{claim}\label{clm:DOLB}[\cite{TZ05}, implicit]
		Assuming Erd\H{o}s girth conjecture, for every (real) $t<2k+1$, there is an $n$-vertex unweighted graph $G=(V,E)$ such that every $(1,t)$-gap-distance oracle has space $\Omega(n^{1+\frac1k})$.
	\end{claim}
	This claim follows implicitly from \cite{TZ05}. We provide a proof for the sake of completeness.
	\begin{proof}		
		The girth of a graph is the length of its shortest cycle. Erd\H{o}s girth conjecture \cite{Erdos64} states that there is a graph $G=(V,E)$ with girth $2k+2$ and $\Omega(n^{1+\frac1k})$ edges.\footnote{Often in the literature the conjecture is referred to as stating that there is a graph $G$ with girth $2k+1$ and $\Omega(n^{1+\frac1k})$ edges. However, as every graph contains a bipartite graph with at least half the edges, the conjecture implies a graph with girth at least $2k+2$.}
		Let $\mathcal{G}$ be all the subgraph of $G$. And consider two different subgraphs $G_1,G_2$. Then there is an edge $\{u,v\}$ that w.l.o.g. belongs to $G_1$ but not $G_2$. It follows that $d_{G_2}(u,v)\ge 2k+1$, as otherwise $G$ will contain a cycle with at most $2k+1$ edges.
		Let $\mathbb{D}_1,\mathbb{D}_2$ be $(1,t)$-gap distance oracles for $G_1,G_2$ respectively. Given the query $(u,v)$, $\mathbb{D}_1$ will return yes, while $\mathbb{D}_2$ will return no. In particular $\mathbb{D}_1\ne \mathbb{D}_2$. As there are $2^{|E|}$ different subgraphs, there are at least $2^{|E|}$ different gap distance oracles. 
		It follows that there is some graph $G'\in\mathcal{G}$, such that the space of every gap distance oracle $\mathbb{D}$ of $G'$, is at least $\log 2^{|E|}=|E|=\Omega(n^{1+\frac1k})$. 
	\end{proof}
	\begin{claim}\label{clm:DecomositionToDO}
		If an $n$-point metric space admits $(t,p,\Delta)$-ThProbabilistic decomposition for $t>1$, then it has a $(\frac\Delta t,\Delta)$-gap distance oracle with $\tilde{O}(n\cdot p)$ space.
	\end{claim}
	\begin{proof}
		Set $s=p^{-1}\cdot2\log n$, and let $\mathcal{P}_1,\mathcal{P}_2,\dots,\mathcal{P}_s$ be random partition drawn from the $(t,p,\Delta)$-ThProbabilistic decomposition.
		Note that all the clusters in all the partitions have diameter at most $\Delta$. Consider a pair $u,v$ such that $d_X(u,v)\le\frac\Delta t$.
		In every partition $\mathcal{P}_i$, $u$ and $v$ belong to the same cluster with probability at least $p$. The probability that in no partition $u$ and $v$  are clustered together is thus at most 
		$(1-p)^s<e^{-ps}=n^{-2}$. By union bound, as there are at most ${n\choose2}$ pairs at distance at most $\frac\Delta t$, with probability at least $\frac12$ every such pair is clustered together in some partition. We thus can assume that we drawn partitions with this property.
		
		For our gap distance oracle we will simply store all the $s$ partitions. Given a query $u,v$, we will return yes iff there is a cluster in some partition containing both $u,v$. Clearly, given a pair $u,v$ such that $d_X(u,v)\le\frac\Delta t$, we will return yes as we insured that there is a cluster in one of the partitions containing both $u,v$. Otherwise, if $d_X(u,v)>\Delta$ then as all the partitions are $\Delta$-bounded, no cluster will contain both $u,v$ and we will return no.
		
		Storing a single partition takes $\tilde{O}(n)$ space, and hence the total size of our gap distance oracle is $\tilde{O}(n\cdot s)=\tilde{O}(n\cdot p)$.
	\end{proof}
	Using \Cref{clm:DOLB} and \Cref{clm:DecomositionToDO}, \Cref{thm:PartitionMetricsLB} follows.
\end{proof}

\section{Applications}
Applying \Cref{obs:PadToProb} on \Cref{cor:PaddedDoubling} and \Cref{thm:StrongMinorFree} we conclude,
\begin{corollary}\label{cor:probabilistic}
	Let $G$ be a weighted graph and $\Delta>0$ some parameter. 
	\begin{itemize}
		\item If $G$ excludes $K_r$ as a minor, it admits a strong $(O(r),\Delta)$-probabilistic decomposition scheme. 		
		\item If $G$ has doubling dimension $\ddim$, it admits a strong $(O(\ddim),\Delta)$-probabilistic decomposition scheme.
	\end{itemize}
\end{corollary}

\subsection{Approximation for Unique Games on Minor Free Graphs}\label{subsec:UG}
In the \emph{Unique Games} problem we are give a graph $G=(V,E)$, an integer $k\ge 1$ and a set of permutations $\Pi=\{\pi_{uv}\}_{uv\in E}$ on $[k]$ satisfying $\pi_{uv}=\pi_{vu}^{-1}$ (each permutation $\pi_{uv}$ is a bijection from $[k]$ to $[k]$). 
Given an assignment $x:V\rightarrow[k]$, the edge  $uv\in E$ is satisfied if $\pi_{uv}(x(u))=x(v)$. 
The problem is to find an assignment that maximizes the number of satisfied edges.
The Unique Games Conjecture of Khot \cite{Kho02} postulates that it is NP-hard to distinguish whether a given
instance of unique games is almost satisfiable or almost unsatisfiable.
The unique games conjecture was thoroughly studied. The conjecture has numerous implications.  

Alev and Lau \cite{AL17} studied a special case of the unique games problem, where the graph $G$ is $K_r$ minor free. Given an instance $(G,\Pi)$ where the optimal assignment violates $\eps$-fraction of the edge constrains, Alev and Lau used an LP-based approach to efficiently find an assignment that  violates at most   $O(\sqrt{\eps}\cdot r)$-fraction.
Specifically, in the rounding step of their LP, they used strong diameter probabilistic decompositions with parameter $O(r^2)$. Using instead our decompositions from \Cref{cor:probabilistic} with parameter $O(r)$ we obtain a quadratic improvement in the dependence on $r$.

\begin{theorem}
	Consider an instance $(G,\Pi)$ of the unique games problem, where the graph $G$ is $K_r$ minor free. Suppose that the optimal assignment violates at most an  $\eps$-fraction of the edge constrains. There is an efficient algorithm that find an assignment that violates at most an $O(\sqrt{\eps\cdot r})$-fraction.
\end{theorem}

\subsection{Spanner for Graphs with Moderate Doubling Dimension}\label{subsec:Spanner}
Given a weighted graph $G=(V,E,w)$, a weighted graph $H=(V,E_H,w_H)$ is a $t$-\emph{spanner} of $G$, if for every pair of vertices $v,u\in V$, $d_G(v,u)\le d_H(v,u) \le t\cdot d_X(v,u)$. 
If in addition $H$ is a subgraph of $G$ (that is $E_H\subseteq E$ and $w_H$ agrees with $w$ on $E_H$) then $H$ is a \emph{graph spanner}.
The factor $t$ is called the \emph{stretch} of the spanner. The number of edges $|E_H|$ is the {\em sparsity} of the spanner. The weight of $H$ is $w_H(H)=\sum_{e\in E_H}w_H(e)$ the sum of its edge weights. The \emph{lightness} of $H$ is $\frac{w_H(H)}{w(\mst(G))}$ the ratio between the weight of the spanner to the wight of the MST of $G$.
The tradeoff between stretch and sparsity/lightness of spanners had been the focus of an intensive research effort, and low stretch graph spanners are used in a plethora of applications.

There is an extensive study of spanners for doubling metrics. Recently, for an $n$-vertex graph with doubling dimension $\ddim$,  Borradaile, Le and Wulff-Nilsen \cite{BLW19} contrasted a graph spanner with $1+\eps$ stretch, $\eps^{-O(\ddim)}$ lightness and  $n\cdot\eps^{-O(\ddim)}$ sparsity (improving \cite{Smid09,Got15,FS20}). This result is also asymptotically tight.
Note that the dependency on $\ddim$ is exponential, which is unavoidable for small, $1+\eps$ stretch. 
In cases where $\ddim$ is moderately large (say $\sqrt{\log n}$), it might be preferable to accept larger stretch in order to obtain small nlightness.  

In a recent work, Filtser and Neiman \cite{FN22}, for every stretch parameter $t\ge 1$, constructed a spanner with stretch $O(t)$, lightness $O(2^{\frac{\ddim}{t}}\cdot t\cdot\log^2n)$ and $O(n\cdot 2^{\frac{\ddim}{t}}\cdot\log n\cdot \log t)$ edges.
However, this spanner was not a subgraph. Most applications require a graphic spanner. It is possible to transform \cite{FN22} into a graphic spanner, but the number of edges becomes unbounded.
The spanner construction of \cite{FN22} is based on a variant of probabilistic decompositions, where they used a weak-diameter version. If we replaced this with our strongly padded decompositions \Cref{cor:PaddedDoubling}, and plug this into Theorem 3 from \cite{FN22}, we obtain a spanner with the same stretch to lightness ratio, but also with an additional sparsity guarantee.
\begin{corollary}
	Let $G=(V,E,w)$ be an $n$ vertex graph, with doubling dimension $\ddim$ and aspect ratio $\Lambda=\frac{\max_{e\in E}w(e)}{\min_{e\in E}w(e)}$. Then for every parameter $t>1$ there is an graph-spanner of $G$ with stretch $O(t)$, lightness $O(2^{\frac{\ddim}{t}}\cdot t\cdot\log^2n)$ and $O(n\cdot 2^{\frac{\ddim}{t}}\cdot\log n\cdot \log \Lambda)$ edges.
\end{corollary}

\subsection{Path Reporting Distance Oracles}\label{subsec:DO}
Given a weighted graph $G=(V,E,w)$, a  {\em distance oracle} is a data structure that supports distance queries between pairs $u,v \in V$. The distance oracle has stretch $t$, if for every query $\{u,v\}$, the estimated distance $\est(u,v)$ is within $d_G(u,v)$ and $t\cdot d_G(u,v)$.
The studied objects are stretch, size the query time.
An additional requirement that been recently studied \cite{EP16} is \emph{path reporting}: in addition to distance estimation, the distance oracle should also return a path of the promised length. In this case, we say that distance oracle has query time $q$, if answering a query when the reported path has $m$ edges, takes $q+O(m)$ time.

Path reporting distance oracles were studied for general graphs \cite{EP16,ENW16}.
For the special case of graphs excluding $K_r$ as a minor, Elkin, Neiman and Wulff-Nilsen \cite{ENW16} constructed a path reporting distance oracles with stretch $O(r^2)$, space $O(n\cdot\log \Lambda\cdot\log n)$ and query time $O(\log\log\Lambda)$, where $\Lambda=\frac{\max_{u,v}d_G(u,v)}{\min_{u,v}d_G(u,v)}$ is the aspect ratio.
For this construction they used the strongly padded decomposition of \cite{AGGNT19} (in fact strong sparse covers). 
Implicitly, given a graph $G$ that admits a strong $(\beta,s)$ sparse cover scheme, \cite{ENW16} constructs a path reporting distance oracle with stretch $\beta$, size $O(n\cdot s\cdot \log_\beta\Lambda)$ and query time $O(\log\log\Lambda)$.
Following similar arguments to  \cite{ENW16} (taking $O(\log n)$ independent copies and using union bound). our padded decompositions from \Cref{thm:StrongMinorFree} implies that every $K_r$ minor free graph admits a strong $(O(r),O(\log n))$ sparse cover scheme. We conclude: 
\begin{corollary}
	Given an $n$-vertex weighted graph $G=(V,E,w)$ which excludes $K_r$ as a minor, with  aspect ratio $\Lambda$ , there is a path reporting distance oracle with stretch $O(r)$, space $O(n\cdot\log_r \Lambda\cdot\log n)$ and query time $O(\log\log\Lambda)$.
\end{corollary}

%It is interesting to mention that Busch \etal \cite{BLT14} constructed a $\left(4,O(f(r)\log n)\right)$ sparse cover scheme for $K_r$ minor free graphs, where $f(r)$ is an extremely large function of $r$. Using the framework of \cite{ENW16}, it will imply a path reporting distance oracle with stretch $4$, space $O(n\cdot\log \Lambda\cdot f(r))$\atodo{The depandance on $\log n$ looks fishy} and query time $O(\log\log\Lambda)$. 
%The value of $f(r)$ is larger than a square of the constant from the Robertson and Seymour structure theorem \cite{RS03}. In particular, an estimation by Johnson \cite{Johnson87} implies that $f(r)$ is larger than $2\Uparrow\left(2\Uparrow\left(2\Uparrow(r/2)\right)+3\right)$\footnote{$2\Uparrow t$ denotes an exponential tower of $t$ $2$'s. That is $2\Uparrow 0=1$ and $2\Uparrow t=2^{2\Uparrow (t-1)}$.}.
%This value is so big, that the \cite{BLT14}-based oracle is completely impractical already for quite small values of $r$.

For the case of graphs with doubling dimension $\ddim$, we constructed the first strong sparse covers. Plugging our \Cref{thm:DdimCover} into the framework of \cite{ENW16}, we obtain the first path reporting distance oracle for doubling graphs. The only relevant previous distance oracle for doubling metrics is by Bartal \etal \cite{BGKLR11}. However, they focused on the $1+\eps$-stretch regime, where inherently the oracle size has exponential dependency on $\ddim$.  

\begin{corollary}
	Given an $n$-vertex weighted graph $G=(V,E,w)$ with doubling dimension $\ddim$  and aspect ratio $\Lambda$, for every parameter $t\ge \Omega(1)$, there is a path reporting distance oracle with stretch $O(t)$, space $O(n\cdot 2^{\nicefrac{\ddim}{t}}\cdot\ddim\cdot\log \Lambda)$ 
	\footnote{This is assuming $\Lambda>\log t$, otherwise simply using an arbitrary shortest path tree will provide a distance oracle with stretch $O(\log t)$.}
	and query time $O(\log\log\Lambda)$.
	
	In particular, there is a path reporting distance oracle with stretch $O(\ddim)$, space $O(n\cdot \ddim\cdot\log \Lambda)$ 	and query time $O(\log\log\Lambda)$.
\end{corollary}

\section{Conclusion and Open Problems}
In this paper we closed the gap left in \cite{AGGNT19} between the padding parameters of strong and weak padded decompositions for minor free graphs.
Our second contribution is tight strong padded decomposition scheme for graphs with doubling dimension $\ddim$, which we also use to create sparse cover schemes.
Finally, we addressed the most basic question of strong and weak ThProbabilistic decompositions for general graphs, optimizing the stretch parameter.
Some open questions remain:
\begin{enumerate}
	\item Prove/disprove that $K_r$ minor free graphs admit strong/weak decompositions with padding parameter $O(\log r)$, as conjectured by \cite{AGGNT19}. 
	\item The question above is already open for the more restricted family of treewidth $r$ graphs.
	\item The $\delta$ parameter: \cite{AGGNT19} constructed weak $\left(O(r),\Omega(1)\right)$-padded decomposition scheme, while we constructed strong $\left(O(r),\Omega(\frac1r)\right)$-padded decomposition scheme. It will be nice to construct strong $\left(O(r),\Omega(1)\right)$-padded decomposition scheme. Such a decomposition will imply a richer spectrum of sparse covers (with $o(r)$ stretch).
	\item Sparse covers for $K_r$ minor free graphs: \cite{AGMW10} constructed strong $(O(r^2),2^r(r+1)!)$ sparse cover scheme, while 
	\cite{BLT14} constructed strong $(4,f(r)\cdot \log n)$ sparse cover scheme. An interesting open question is to create additional sparse cover schemes. Specifically, our padded decompositions suggest that an $(O(r),g(r))$-strong sparse cover scheme for some function $g$ independent of $n$, should be possible. Currently it is unclear how to construct such a cover. Optimally, we would like to construct $(O(1),g(r))$-sparse cover scheme.
	\item Close the gap between the lower and upper bounds of weak ThProbabilistic decompositions for general metrics.
	Specifically show that general $n$-point metrics admit $(2k-1,n^{-\frac1k})$-ThProbabilistic decomposition scheme, or the impossibility of such a scheme.
	\item Close the gap between weak to strong stochastic decompositions for general graphs.
	Specifically, construct strong $(2k,n^{-\frac1k})$-stochastic decompositions scheme, or show the impossibility of such a scheme.
\end{enumerate}

\section{Acknowledgments}
The author would like to thank Ofer Neiman for helpful discussions.

{\small
	\bibliographystyle{alphaurlinit}
	\bibliography{SteinerBib}
}
\newpage
\appendix

\end{document}